\documentclass[lettersize,journal]{IEEEtran}
\usepackage{amsmath,amsfonts}
\usepackage{algorithmic}
\usepackage{algorithm}
\usepackage{array}
\usepackage{textcomp}
\usepackage{stfloats}
\usepackage{url}
\usepackage{verbatim}
\usepackage{cite}
\usepackage{sharina}
\usepackage{setspace}
\usepackage{amsthm}
\usepackage{amssymb,amsmath,latexsym}
\usepackage{graphicx,subfigure}
\usepackage{color,cite,times}
\usepackage{epstopdf}
\usepackage{multirow}
\usepackage{array}

\newtheorem{assumption}{Assumption}

\newtheorem{lemma}{\textbf{Lemma}}
\newtheorem{corollary}{\textbf{Corollary}}
\newtheorem{proposition}{\textbf{Proposition}}
\newtheorem{remark}{\textbf{Remark}}
\allowdisplaybreaks[4]

\begin{document}
	\title{UAV-Enabled Asynchronous Federated Learning}
	\author{	Zhiyuan~Zhai, 
		Xiaojun~Yuan,~\IEEEmembership{Senior Member, IEEE}, 
		Xin~Wang,~\IEEEmembership{Fellow, IEEE}, and Huiyuan~Yang,~\IEEEmembership{Graduate Student Member, IEEE}
	}
	\pdfoutput=1
	\maketitle

	\begin{abstract}
	To exploit unprecedented   data generation in  mobile edge networks, federated learning (FL) has emerged as a promising alternative to the conventional centralized machine learning (ML). By collectively training  a unified learning model on edge devices, FL bypasses the need of direct data transmission, thereby addressing problems such as latency issues and privacy concerns inherent in centralized ML. 
	However, FL still faces  some critical challenges in deployment.
	One major challenge called straggler  issue severely limits FL's  coverage where the device with the weakest channel condition  becomes  the bottleneck of the model aggregation performance.
	Besides, the huge uplink communication overhead  compromises the effectiveness of FL, which is particularly pronounced in large-scale systems. 
	To address  the straggler issue, we propose the integration of an unmanned aerial vehicle (UAV) as the parameter server (UAV-PS) to coordinate the  FL implementation. 
	We further employ over-the-air computation technique that leverages the superposition property of wireless channels for efficient uplink communication. 
	Specifically, in this paper, we  develop a novel UAV-enabled over-the-air asynchronous FL (UAV-AFL) framework which supports the  UAV-PS in updating the model continuously  to enhance the learning performance. Moreover, we introduce  a `staleness upper bound' metric  to control the asynchronous level in AFL and conduct a convergence analysis  to quantitatively capture the impact of  model asynchrony, device selection and communication errors on the UAV-AFL learning performance. Based on this, a unified communication-learning  problem is formulated to maximize asymptotical learning performance by optimizing the UAV-PS trajectory, device selection and over-the-air transceiver design. Simulation results reveal valuable insights into the UAV-AFL  system  and demonstrate that the proposed UAV-AFL scheme achieves substantially learning efficiency improvement compared with the state-of-the-art approaches.
\end{abstract}


\section{Introduction}
Extensive data available on mobile edge devices have sparked considerable interest in the development of artificial intelligence (AI) services, such as image recognition  and natural language processing \cite{he2016deep}, \cite{young2018recent}. However, traditional machine learning (ML) requires the collection of data in a centralized server for model training, which incurs substantial energy and bandwidth costs, significant time delays, and potential privacy concerns \cite{zhu2020toward}.
Federated learning (FL), as a novel ML paradigm,  has emerged to tackle these problems \cite{konevcny2016federated}. In a typical FL framework, each edge device computes local model updates based on its individual dataset and subsequently transmits the model update to a parameter server (PS) for model aggregation. The global model is then  updated at the PS and shared with the participating devices. In this way, FL replaces direct data transmission with the model updates (gradients) transmission, which not only alleviates the communication burden but also safeguards the confidentiality of local data.

Despite the above advantages, the huge uplink communication overhead incurred in the iterative learning process remains a significant  bottleneck of FL \cite{konevcny2016federated}.   Over-the-air computation  \cite{zhu2019broadband,zhu2018mimo}
for FL uplink model aggregation  is a promising technique to reduce this  communication cost.  By using the superposition property of electromagnetic waves in wireless channels, edge devices can concurrently transmit their gradients utilizing the same time-frequency resources.
The authors in \cite{zhu2019broadband} thus  introduce  over-the-air model aggregation  for FL, and the result demonstrates  substantial communication efficiency improvement and latency reduction when compared to orthogonal multiple access (OMA) schemes. This approach  is  extended to multiple-input multiple-output (MIMO) scenarios in \cite{zhu2018mimo}
 where  over-the-air computation is applied.


In addition to the communication burden, FL also encounters a new design challenge, namely the straggler issue  \cite{9451567,zhong2022over}. 
Since the PS needs to aggregate gradients from all participating  devices, the devices with poor channel conditions (i.e., stragglers) always introduce huge communication errors into the aggregated results.
This is particularly true when FL is deployed over a wide range of distributed devices where the   errors from distant stragglers severely  impair the overall aggregation performance.
In this context, using an unmanned aerial vehicle (UAV) as PS (UAV-PS) is a promising approach to provide FL services. This is because the UAV has the ability to move closer to the stragglers and ensure reliable communication links with them.
The enhanced capabilities, such as superior communication quality and expanded  coverage can be attained by utilizing the high mobility and on-demand deployment ability of UAV \cite{wu2019fundamental}, \cite{hua2018power}.
For this reason, the authors in \cite{lim2021uav} propose to deploy a UAV  as a network relay, where the UAV  collects model updates from devices and then forward them to the model owner for improving the FL coverage. Moreover, in order to minimize a fairness metric regarding device computing time, a UAV-PS trajectory optimization is carried out in \cite{donevski2021federated}.

Till now, the UAV-assisted FL system using over-the-air computation technique  has only been investigated in \cite{zhong2022uav} and \cite{10283588}. In particular, \cite{zhong2022uav} introduces a hierarchical aggregation approach where UAV-PS collects gradients from devices at various locations and  merges them at the end of each iteration.  \cite{10283588} dispatches the UAV-PS to assist with the aggregation of different learning models and determines the transmission scheduling using  greedy deflation approach. However, these studies \cite{lim2021uav,donevski2021federated,zhong2022uav,10283588}, while successfully integrating UAVs into the FL system to enhance learning performance, have not addressed a pivotal issue: 
The synchronous structure of FL is inherently incompatible with the asynchronous nature of UAV-assisted communication systems. Specifically, in the FL system, the  model aggregation  is only finished when all the gradients of devices are collected. 
However, the UAV-PS can only   collect the gradients of the nearby devices (i.e., a subset of all devices). 
This procedural incompatibility leads to a significant delay in model update until the UAV-PS  navigates through the entire service area and collects gradients from all  devices,  which severely impairs the  overall efficiency of the learning process.  
A feasible solution  to mitigate this incompatibility is to allow the UAV-PS to execute one  FL iteration exclusively with the nearby devices. In this approach, the UAV-PS broadcasts the global model to these adjacent  devices, and the latter then use the model for gradient computation. Subsequently, the UAV-PS collects the computed gradients and updates the model parameters. After that, the UAV-PS moves to  another  area and begins new FL iteration with the nearby devices.
However, while this scheme reduces the latency of model updates, it requires the UAV-PS  to hover over the serviced devices when they conduct local gradient computation. This time-consuming local computation process\footnote{In the internet of things (IoT) networks, edge devices often require several seconds for completing a single gradient descent (GD) algorithm due to their limited CPU capacity\cite{8737464},\cite{li2018adaptive}, which is significantly greater than the gradient transmission time.}
 could also adversely impact the overall learning efficiency.

To fully unleash the system potential and maximize the learning efficiency, we develop a novel UAV-enabled over-the-air asynchronous FL (AFL) scheme. This scheme is generally obtained by further changing the order of
 the model aggregation and  broadcast steps in the framework of aforementioned solution.
Since the aggregation  happens before the broadcast step, some aggregated gradients  are  computed based on different versions of the global model, which leads to an AFL problem. 
Specifically, in the proposed UAV-enabled over-the-air AFL (UAV-AFL) scheme,  the UAV-PS selects some devices to communicate with,  collects these  devices' gradients over-the-air,   updates the  model  and then performs  model broadcasting to these selected devices. After that, these selected devices conduct local computation based on the newly received model while the UAV-PS continues the next round of AFL by selecting new devices without waiting for the  local computations of these previously selected devices. 
Compared with conventional UAV-assisted FL settings \cite{lim2021uav,donevski2021federated,zhong2022uav,10283588},
the UAV-AFL scheme improves learning efficiency by offering a compatible asynchronous communication-learning design which allows the UAV-PS to traverse the service area to collect gradients and update models  continuously.

In this paper, we investigate the potential of the UAV-AFL system by developing a unified analysis framework that captures the impact of model asynchrony, device selection and communication errors on the AFL training loss. Based on this, a joint communication-learning  optimization algorithm is formulated to maximize the asymptotic learning performance. Our main contributions  are summarized as follows.
\begin{itemize}
	\item  We study a UAV-assisted  FL system where a UAV-PS is deployed to perform  over-the-air aggregation over line-of-sight (LoS) channels.  We propose a novel UAV-AFL scheme to improve the learning efficiency. This scheme  has two  advantages: 1) The UAV-PS can immediately updates the global model upon the reception of  the aggregated gradients. 2) The UAV-PS can collect the newly computed gradients for model update when  devices are conducting  local computation.
	\item  
	We introduce a metric called `staleness upper bound', i.e., the maximal version discrepancy among the models maintained by all edge devices,  to characterize the level of asynchrony in the UAV-AFL system.  Then,  a rigorous  convergence analysis on the  learning loss is carried out  by taking the model asynchrony, device selection  and communication errors into account. To the best of our knowledge, this is the first attempt to analyze the convergence of asynchronous learning/optimization in the presence of general communication error and device selection loss.
	\item  We formulate a  communication-learning  co-design optimization problem to enhance the learning performance of the UAV-AFL based on the convergence analysis.
An effective algorithm is developed to solve the  problem. Specifically, we transform the original problem via geometric programming (GP), and  design a two-layer penalty-based algorithm using  successive convex
	approximation (SCA) principle to jointly optimize the UAV trajectory, device selection and transceiver design. 
\end{itemize}
Simulation results under different system settings provide an in-depth study of the trade-offs between over-the-air aggregation quality, model update rates and asynchronous level, which offers valuable insights for the   UAV-AFL system design.	The comparison results indicate that the proposed UAV-AFL scheme achieves a substantial improvement in learning efficiency compared to the existing state-of-the-art solutions.

\textit{Notations:}
We denote the real and complex number sets by $\mathbb{R}$ and $\mathbb{C}$, respectively.  We use  $(\cdot)^\mathrm{T}$ to denote the  transpose; $|\cS|$ to denote the cardinality of set $\cS$; $\circ$ to denote the hadamard product; $\E$ to denote the expectation operator; $\norm{\cdot}$ to denote the $l_2$-norm; $\mathcal{CN}(\mu,\sigma^2)$ to denote circularly-symmetric complex normal distribution with mean $\mu$ and covariance $\sigma^2$.  
	\section{System Model}
	
	\subsection{Federated Learning System}
	We consider a  FL system with $M$ devices collaboratively training a common global  model, which can be characterized by a $D$-dimensional model parameter $\xx \in \mathbb{R}^{D}$.  The entire training data set $\mathcal{S}$, with size $\left|\mathcal{S}\right|$, is distributed over the edge devices. We use $\cM=\left\{1,\cdots,M\right\}$ to denote the set of  devices involved in the system. Suppose that each device possesses a local training data set $\mathcal{S}_m$ with  size $\left|\mathcal{S}_m\right|=S_d$, where $\sum_{m\in\cM}\left|\mathcal{S}_m\right|=MS_d=\left|\mathcal{S}\right|$ and $\mathcal{S}_i\cap\mathcal{S}_j=\varnothing,\forall i \neq j$. The learning objective of the FL system is to find an optimal parameter $\xx^{\star}$ minimizing the empirical loss function given by
	\begin{align}\label{overall_obj_fl}
		F(\xx)=\frac{1}{\left|\mathcal{S}\right|}\sum_{\xi_i\in \cS}f(\xx,\xi_i),
	\end{align}
	where $\xi_i$ is the $i$-th training sample including the input feature and output label, and $f(\xx,\xi_i)$ is the  loss function with respect to model $\xx$ based on sample $\xi_i$.  Denote the local loss function of device $m$ as 
	\begin{align}
		f_m(\xx)=\frac{1}{\left|\mathcal{S}_m\right|}\sum_{\xi_i \in \cS_m}f(\xx,\xi_i).
	\end{align}
	Then, the 	FL task in \eqref{overall_obj_fl} can be represented as 
	\begin{align}
		\min_{\xx \in \mathbb{R}^d} ~F(\xx)=\frac{1}{M}\sum_{m\in \cM}f_m(\xx).
	\end{align}

	\subsection{UAV-Assisted Communication System}
	A UAV-assisted communication network is used to fulfill the FL task, where a UAV-operated parameter server (i.e., UAV-PS) is deployed to coordinate the edge devices performing  collaborative training. We employ a three-dimensional Cartesian coordinates system, where the coordinates of the $m$-th ground device   are $\ww_m=\left[x_m,y_m,0\right], \forall m$. We assume that the UAV-PS maintains a constant flight altitude $H$ and the position of the UAV-PS  remains constant  during one time slot with duration  $\delta_t$. The coordinates of the UAV-PS at time slot $k$, denoted by $\qq(k)$, can be expressed as $\qq(k)=\left[x(k),y(k),H\right]$. We then have the following mechanical constraints
	\begin{align}
		&\vv(k)=\frac{\qq(k+1)-\qq(k)}{\delta_t},\norm{\vv(k)}\leq v_{\max}, \forall k, \label{mech_1}\\
		&\aa_k=\frac{\vv(k+1)-\vv(k)}{\delta_t},\norm{\aa_k}\leq a_{\max},\forall k,\label{mech_2}\\
		&\qq(0)=\qq(K)=\qq_F,\label{mech_3}
	\end{align}
	where $v_{\max}$ (or $a_{\max}$) denotes the maximal speed (or acceleration) of the UAV-PS, $K$ is the total number of time slots, and $\qq_F$ denotes the initial and final dispatch position of the UAV-PS.
	
	We assume  that the channel state information (CSI) is available at edge devices and the UAV-PS \cite{cao2020cooperative}, and  the Doppler effect arising from the movement of the UAV-PS can be effectively compensated at the receiver. Moreover, due to the ground-to-air nature of UAV communication, we focus on the line-of-sight (LoS) channels for the connection links. 
	Therefore, the channels connecting the devices to the UAV-PS adhere to the free-space path loss model, expressed as
	\begin{align}\label{channel}
		h_m(k)=\sqrt{g_0d_m^{-2}(k)}\vartheta_m(k),m \in {\cN}_k, \forall k 
	\end{align}
	where $h_m(k)$ is the channel between the $m$-th device and the UAV-PS in the $k$-th time slot, $g_0$ represents the channel power gain at the reference distance $d_\text{ref}=1 \text{m}$, $d_m(k)=\norm{\qq(k)-\ww_m}$ denotes the distance between device $m$ and the UAV-PS at the $k$-th time slot,  $\vartheta_m(k)=e^{j\theta_m(k)}$ with $\theta_m(k)$ being the phase shift of $h_m(k)$, and ${\cN}_k$ denotes  the set of devices  that have  LoS communication links with the UAV-PS during time slot $k$.

	\section{UAV-Enabled Over-the-Air AFL Framework}\label{sec33}
	\subsection{Motivation}
	In the conventional UAV-assisted FL literature \cite{lim2021uav,donevski2021federated,zhong2022uav,10283588}, the system conducts the following steps iteratively: 1) The UAV-PS broadcasts the global model to all  devices. 2) All devices conduct local gradient computation based on the received model.  
	3) The UAV-PS  navigates through the entire service area and collects gradients from all  devices.
	4) The UAV-PS finally aggregates the gradients and use it to update the global model.
	Here, the UAV-PS can only update the global model after all devices have completed local computations and the UAV-PS  has finalized the gradient collection of all devices (i.e., after step 3 and 4). Hence, this approach needs  excessively long time to carry out a single model update, which significantly impairs learning efficiency and effectiveness.
	
	To  address this issue, a straightforward choice  is to allow each FL iteration to involve only the devices located near the UAV-PS. In this method, the UAV-PS   traverses the entire service area and  provides FL service sequentially to the nearby devices.  When the UAV-PS approaches some devices, the system conducts the following steps:
	1) The UAV-PS  broadcasts the global model to the near devices. 2) These devices compute the gradients locally. 3) The UAV-PS collects the newly computed gradients from these devices. 4) The UAV-PS aggregates the gradients and updates the global model.
	In this way, the UAV-PS  can directly utilize the gradients from  near  devices to update the global model, thereby enhancing the learning efficiency. However, to collect the gradients (step 4), the UAV-PS in this approach cannot do anything but  hover near the devices to wait for them completing local computations, which  also compromises the system efficiency.
	
	Therefore, to improve the overall learning efficiency, it is crucial to develop a strategy where the UAV-PS can operate without interruption. This strategy should enable the UAV-PS to collect gradients and update models continuously for ensuring a seamless and efficient learning process.
	\subsection{Proposed UAV-Enabled AFL System}\label{sec3subA}
	\begin{figure}[h]
		\centering
		\includegraphics[width=0.3\textwidth]{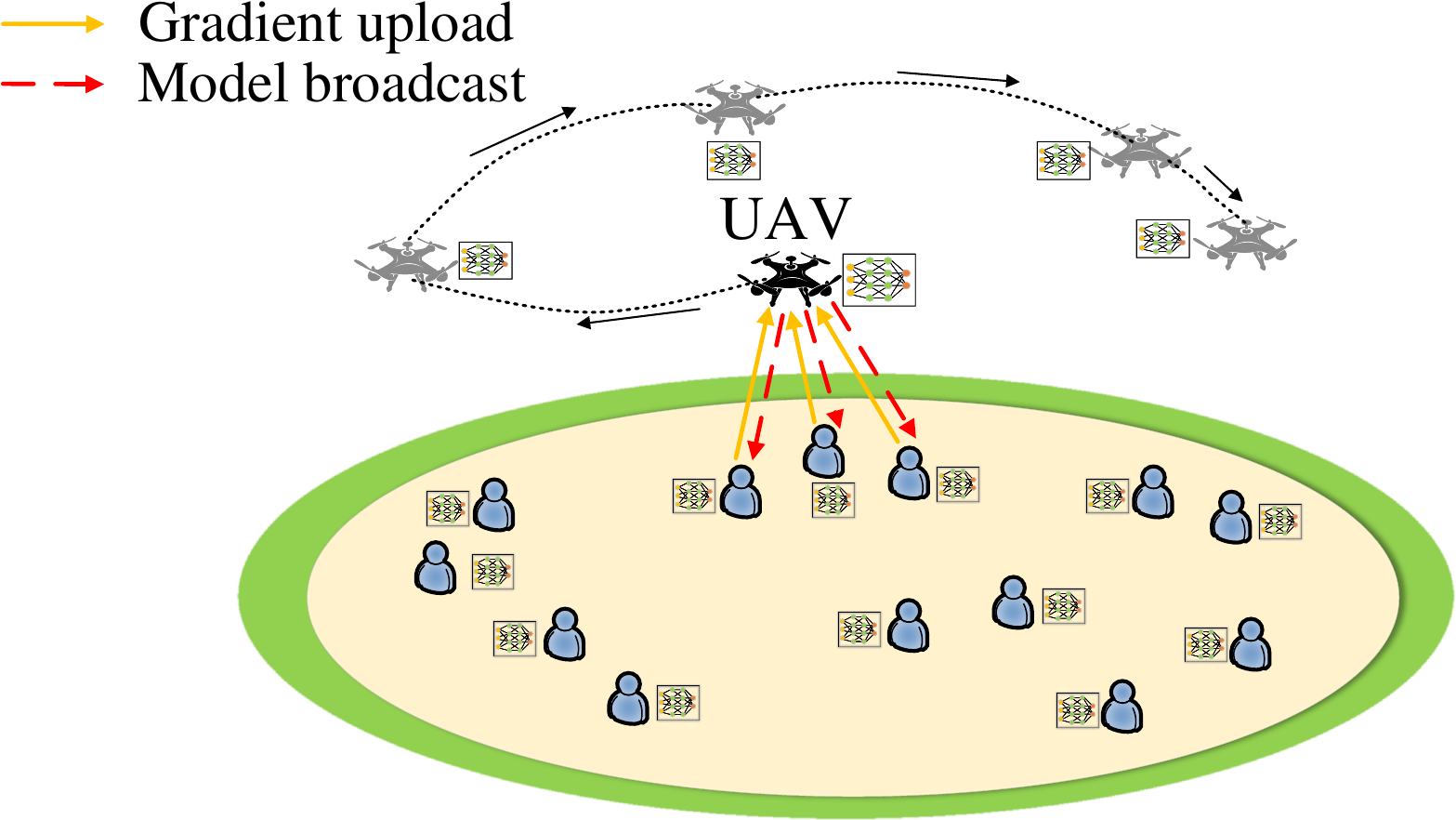}
		\caption{UAV-enabled AFL system model.}
		\label{fig_system_model}
	\end{figure}
	We consider a UAV-enabled  AFL  system as depicted in Fig.~\ref{fig_system_model}. 
In this system, at each time slot,  the UAV-PS selects a subset of devices to communicate with. Specifically, the UAV-PS collects the gradients of the selected devices, updates the global model, and then broadcast the updated model to these devices.  Subsequently, the selected devices initiate computation based on the received model, while the UAV-PS continues to  collect gradients from other devices for updating the model at the same time. 
	Denote by $\cJ_k \subset \cM $  the set of aggregation-ready devices in the $k$-th time slot, i.e., the devices that complete their local computation and are ready for aggregation at time slot $k$.
	In the $k$-th time slot, the system conducts the following procedures:
	
	\begin{itemize}
		\item \textit{Device selection}: 
		The UAV-PS selects a subset of devices in $\cJ_k \cap  \cN_k$ to participate in the aggregation process. Let  $\cM_k$ be the set of selected devices in the $k$-th time slot. We define $a_m(k)$ as an indicator function of the  selection status of the \(m\)-th device in the \(k\)-th time slot. That is, $a_m(k)$ = 1 if $m \in \cM_k$, and $a_m(k)$ = 0 otherwise.
		\item \textit{Gradient aggregation}: 
		The  devices in  $\cM_k$ transmit the  computed  gradients $\gg_{m}(k), m \in \cM_k$ to the UAV-PS via the wireless channels.
		Here, $\gg_{m}(k)$ is represented as 
		\begin{align}\label{gra_diver}
			\gg_{m}(k)=\nabla f_m(\xx(k-\tau_{m,k})), m \in \cM_k,
		\end{align}
		where  $\xx(k-\tau_{m,k})$ denotes  the global model of time slot $k-\tau_{m,k}$,  $\nabla f_m(\xx(k-\tau_{m,k}))$ is the gradient of device $m$ with respect to the model $\xx=\xx({k-\tau_{m,k})}$, and  $~\tau_{m,k}=k-\max_{k'<k}\left\{k'\mid m \in \cM_{k'}\right\}$ is referred to as \emph{staleness} \cite{lian2015asynchronous}  of device $m$ in time slot $k$.
		Guided by the received signals,
		the UAV-PS obtains an estimate  $\hat{\gg}(k)$ of the desired aggregated gradient $\gg(k)$. The desired gradient $\gg(k)$ is given by
		\begin{align}\label{de_agg}
			\gg(k)=\frac{1}{\left|\cM_k\right|}\sum_{m\in \cM_k}\gg_m(k).
		\end{align} 
		However, due to the impact of channel fading and communication noise, we in general have $\hat{\gg}(k)\neq\gg(k) $. With $\hat{\gg}(k)$,   the global model $\xx(k)$ is then updated by 
		\begin{align}\label{update}
			\xx(k+1)=\xx(k)-\lambda\hat{\gg}(k),
		\end{align}
		where $\lambda$ is the learning rate.
		\item \textit{Model broadcast}: 
		The UAV-PS broadcasts the updated global model $\xx(k+1)$ to the  devices $m\in \cM_k$. 
		\item \textit{Local computation}: 
		The devices $m\in \cM_k$  compute their individual gradients using $\xx(k+1)$. Let $c_m$ represent the computation time of the $m$-th device. Hence, the devices $m\in \cM_k$ cannot communicate with the UAV-PS from time slot $k+1$ to $k+c_m$ inclusive.
	\end{itemize}
	\begin{remark}\label{remark1}
		In the synchronous FL setting \cite{chen2020joint,amiri2020federated,yang2020energy}, all the gradients $\gg$'s are guaranteed to be computed using the unified global model $\xx(k)$. However, in the UAV-enabled  AFL  system, some aggregated gradients $\gg_{m}(k)$  might be computed based on the stale global model, i.e., the  model $\xx(k')$, where $k'<k$.
	\end{remark}

		\begin{remark}	
		Compared to the approaches in\cite{lim2021uav,donevski2021federated,zhong2022uav,10283588}, the UAV-AFL scheme  provides  a more efficient learning strategy. Here, the UAV-PS can seamlessly collect the newly computed gradients from edge devices and update the model parameters immediately.
		 This feature   not only accelerates the UAV-PS  update rate but also increases the device computing efficiency.
		\end{remark}

	\subsection{Over-the-Air Model Aggregation}
	We employ the over-the-air computation technique \cite{zhu2019broadband} to achieve efficient   aggregation in the UAV-AFL system. Specifically, in the $k$-th time slot, the selected devices  $m\in \cM_k$ collaboratively transmit their gradients $\{\gg_{m}(k) : m \in \cM_k\} $  using the same time-frequency resources. By  manipulating the transmit and receive equalization coefficients, the estimated weighted sum $\hat{\gg}(k)$ in \eqref{update} can be coherently superposed at the UAV-PS. The specifics are elaborated as follows.
	
	To perform over-the-air aggregation, the initial step is the normalization of $\gg_{m}(k),m \in \cM_k$, which can be expressed as
	\begin{align}\label{normalization}
		\tilde{\gg}_{m}(k)=\left(\gg_{m}(k)-\bar{g}_{m}(k)\1_{D\times1}\right)/\sqrt{v_{m}(k)},
	\end{align}
	where $\tilde{\gg}_{m}(k) \in \mathbb{R}^D$ is the normalized  zero-mean and unit-variance version of $\gg_{m}(k)$,  and $\bar{g}_{m}(k)=\frac{1}{D}\sum_{d=1}^{D}{g}_{m}^{(k)}(d)$ as well as $v_{m}(k)=\frac{1}{D}\sum_{d=1}^{D}\left({g}_{m}^{(k)}(d)-\bar{g}_{m}(k)\right)^2$ are the mean and variance of $\gg_{m}(k)$, respectively. By following the common practice as in \cite{9451567} and \cite{lin2021deploying}, we assume $\left\{\bar{g}_{m}(k),v_{m}(k)| m\in \cM_k\right\}$ to be transmitted to the UAV-PS via error-free links. We then modulate (analog modulate)  $\tilde{\gg}_{m}(k)$  as a complex version $\rr_{m}(k)\in \mathbb{C}^C$ for transmission, i.e.,
	\begin{align}\label{module}
		\rr_{m}(k)=\tilde{\gg}_{m}(k)\Big(1:\frac{D}{2}\Big)+j\tilde{\gg}_{m}(k)\Big(\frac{D+2}{2}:D\Big)
	\end{align}
	where we assume  $D$ is even  and $C=\frac{D}{2}$. Let $\beta_{m}(k)$ denote the pre-processing factor of device $m$ at time slot $k$. Then the transmit signal $\mathbf{s}_{m}(k)$ can be expressed as
	\begin{align}\label{smk}
		\mathbf{s}_{m}(k)=\beta_{m}(k)\rr_{m}(k)\in \mathbb{C}^C,
	\end{align}
	and the corresponding power constraint is $\E{\norm{s_{m}^{(k)}[l]}^2}=2|\beta_{m}(k)|^2\leq P_0$, where $s_{m}^{(k)}[l]$ is the transmit signal of device $m$ at the $l$-th channel use. To compensate  the phase shift of channel $h_{m}(k)$, we  construct  $\beta_{m}(k)$ as $\beta_{m}(k)={b_{m}(k)}e^{-j\theta_m(k)}$, where $\theta_m(k)$ is the phase of  $h_{m}(k)$ and $b_m(k)\in \mathbb{R}$ is amplitude of $\beta_{m}(k)$.
	By such design, the received signal at the UAV-PS, denoted by $\yy(k)$,  is given by
	\begin{align}\label{rece1}
		\yy(k)=\sum_{m\in\cM_k}|h_m(k)|{b_m(k)}\rr_{m}(k)+\nn(k),
	\end{align}
	where $\nn(k) \in \mathbb{C}^C$ is an additive white Gaussian noise (AWGN) vector  with the elements following the distribution  $\mathcal{C}\mathcal{N}\left(0,\delta_n^2\right)$. Subsequently, the UAV-PS employs  a denoising factor  $\zeta(k)\in \mathbb{R}$ to process the received signal $\yy(k)$ as
	\begin{align}\label{denoising}
		\hat{\rr}(k)\!=\!\zeta(k)\yy(k)=\!\zeta(k)\!\bigg(\!\sum_{m\in\cM_k}\!\!|h_m(k)|{b_m(k)}\rr_{m}(k)\!+\!\nn(k)\!\bigg)\!,
	\end{align}
	and then the estimate $\hat{\gg}(k)$  is finally reconstructed as
	\begin{align}\label{demodule}
		\hat{\gg}(k)=\left[\text{Re}\left\{\hat{\rr}(k)\right\}^{\T},\text{Im}\left\{\hat{\rr}(k)\right\}^{\T}\right]^{\T}+\bar{g}(k)\1_{D\times1} 
	\end{align}
	where $\bar{g}(k)=\frac{1}{|\cM_k|}\sum_{m\in \cM_k}\bar{g}_m{(k)}$ is introduced to compensate the subtracted mean of gradients $\bar{g}_m(k), m \in \cM_k$ in the normalization step \eqref{normalization}.
	
	\subsection{Overall UAV-AFL Framework}
	We summarize the UAV-enabled over-the-air AFL (UAV-AFL) framework discussed above in Algorithm \ref{alg_FL_framework}. We define the  device selection, UAV-PS trajectory, device transmit amplitude gain  and the UAV-PS denoising factor over $K$ time slots as  $\mA = \{ a_{m}(k) \mid\forall m,k \}\in \mathbb{R}^{M\times K}$, $\mQ =\{\qq(k)\mid \forall k \} \in \mathbb{R}^{3\times K}$, $\mB= \{ b_m(k) \mid \forall m,k \} \in \mathbb{R}^{M\times K}$, and $\mZ=\{\zeta(k)\mid \forall k\}\in \mathbb{R}^{K}$, respectively.
	
	\begin{algorithm}
	\caption{UAV-AFL Framework} 
	 	\label{alg_FL_framework} 
	 	\begin{algorithmic}[hp] 
		 		\STATE {\textbf{Initialization:}} The UAV-PS  distributes the initial model $\{\mathbf{x}(1)\}$ to each device and  optimizes $(\mA, \mQ, \mB,\mZ)$ based on the estimated CSI and  device state information.
		 		\FOR{ $k \in [K]$ }
		 		\STATE The  devices in $\cM_k$ transmit  $\{\bar{g}_{m}(k),v_{m}(k)\}_{ m\in \cM_k}$ to UAV-PS through error-free links.
		 		\STATE The devices in $\cM_k$ transmit  $\{\gg_m(k)\}_{ m\in \cM_k}$ to the UAV-PS  based on \eqref{normalization}-\eqref{rece1}.
		 		\STATE The UAV-PS recovers the aggregated gradient $\hat{\gg}(k)$ based on \eqref{denoising}-\eqref{demodule}.
		 		\STATE The UAV-PS updates $\xx(k)$ based on \eqref{update} and broadcasts it to the devices in $\cM_k$.
		 		\STATE  The devices in $\cM_k$ carry out computations based on \eqref{gra_diver} and then wait for the next round of selection.
		 		\ENDFOR 
		 	\end{algorithmic}
		\end{algorithm}

	
	\section{Convergence Analysis}\label{convergence}
	
	In this section, we analyze the convergence performance of the proposed UAV-AFL system.  We present several  assumptions in Section \ref{convergence}-A and derive  a tractable convergence bound by quantitatively characterizing the impact of $\{\mA, \mQ, \mB,\mZ\}$ on the learning performance  in Section \ref{convergence}-B.
	\subsection{Assumptions}

	To conduct the convergence analysis, we make the following assumptions.
	
	\begin{assumption}\label{as1}
		\rm{(Lipschitzian gradient)}. The function $F(\cdot)$ is differentiable and the corresponding gradients $\nabla F(\cdot)$ is Lipschitz continuous with parameter $L$, i.e.,
		\begin{align}
			\norm{\nabla F(\xx)-\nabla F(\yy)}\leq L \norm{\xx-\yy}, \forall \xx,\yy.
		\end{align}
	\end{assumption}
	\begin{assumption}\label{as2}
		\rm{(Strongly convex)}. The function $F(\cdot)$ is strongly convex with parameter $\mu$, i.e.,
		\begin{align}
			F(\xx)\geq F(\yy)\!+\!(\xx-\yy)^T\nabla F(\yy)\!+\!\frac{\mu}{2}\norm{\xx-\yy}^2, \forall \xx,\yy.
		\end{align}
	\end{assumption}
	\begin{assumption}\label{as3}
		\rm{(Staleness upper bound)}. All staleness   $\tau_{m,k},\forall m,k$ are bounded with the value  $S$, i.e., $\tau_{m,k}\leq S, \forall m, k $.
	\end{assumption}
	\begin{assumption}\label{as4}
		\rm{(Bounded gradient)}. The gradient  $\nabla F(\xx_{k-a_{m,k}})$ with the staleness  $a_{m,k}\leq 2S$ is bounded by $\nabla F(\xx_{k})$ with parameters $\alpha_1$ and $\alpha_2$, i.e.,
		\begin{align}\label{asss4}
			&\norm{\nabla F(\xx({k-a_{m,k}}))}^2\notag\\&\quad \leq \alpha_1+\alpha_2\norm{\nabla F(\xx(k))}^2,\forall a_{m,k}\leq 2S, m ,k.
		\end{align}
	\end{assumption}
	\begin{assumption}\label{as5}
		\rm{(Device variance)}. The variance of local gradient and the  global counterpart  is  bounded  with parameters $\delta$, i.e.,
		\begin{align}
			\norm{\nabla f_m(\xx)\!-\!\nabla F(\xx)}^2\leq \delta^2,\forall m, k.
		\end{align}
	\end{assumption}
	
	Assumptions \ref{as1}, \ref{as2} and \ref{as5} are   widely used in the stochastic optimization literature; see e.g., \cite{chen2020joint}, \cite{friedlander2012hybrid}, \cite{bertsekas1995neuro}.
	Assumption \ref{as1} enables us to choose an appropriate learning rate  for balancing convergence speed and learning stability. Assumption \ref{as2} ensures the existence of a unique global optimum $\xx^{\star}$ for the loss function $F(\cdot)$. Assumption \ref{as5}  limits the discrepancy between the gradients of the local and  global loss functions. 
	
	Assumptions \ref{as3} and \ref{as4} are introduced because of the asynchronous nature of UAV-AFL. Assumption \ref{as3} is commonly employed in the analysis of asynchronous algorithms, e.g., \cite{6877255}, \cite{liu2014asynchronous}. 
	Moreover, since the model version gap (staleness) is bounded, we can always find suitable values for $\alpha_1$ and $\alpha_2$ ensuring \eqref{asss4} holds.

%
	\subsection{Convergence  Analysis}
	By comparing with the ideal model update case, we rewrite \eqref{update} as
	\begin{align}\label{update_error}
		\xx(k+1)\!=\!\xx(k)\!-\!\lambda\hat{\gg}(k)\!=\!\xx(k)\!-\!\!\lambda\!\left(\nabla F(\xx(k))-\ee(k) \right)
	\end{align} 
	where $\nabla F(\xx(k))=\frac{1}{\left|\mathcal{S}\right|}\sum_{\xi_i\in \cS}\nabla f(\xx(k),\xi_i)$ is the gradient of $F(\xx)$ at $\xx=\xx(k)$, and $\ee(k)$ represents the overall error introduced  during the gradient evaluation process. 
	Under the UAV-AFL framework,  $\ee(k)$  originates from the following three aspects:
	\begin{itemize}
		\item \textit{Communication error $\ee_{c}(k)$}: The inherent communication noise   introduces  deviations into aggregated gradient $\hat{\gg}(k)$, consequently undermining the learning accuracy.
		\item \textit{Device selection error $\ee_{d}(k)$}: Given the common case that $\cM_k\neq\cM$, the aggregated gradients are computed using a subset of data set $\cS$, which can potentially result in training imbalance and overfitting \cite{hu2023scheduling}. 
		\item \textit{Model Asynchrony error $\ee_{a}(k)$}: As shown in \eqref{gra_diver}, the aggregated gradients are evaluated based on different model parameters, so the model  inconsistency exists within each aggregation process. 
	\end{itemize}
	We  mathematically formulate the above analysis of $\ee(k)$   as
	\begin{align}
		&\ee(k)=\underbrace{\nabla F(\xx(k))-\frac{1}{M}\sum_{m\in\cM}\nabla f_m(\xx({k-\tau_{m,k}}))}_{ \ee_{a}(k)}\notag\\&\!+\!\!\underbrace{\frac{1}{M}\!\!\!\sum_{m\in \cM}\!\!\!\nabla f_m(\xx({k-\tau_{m,k}}))\!-\!\frac{1}{|\cM_k|}\!\!\sum_{m\in \cM_k}\!\!\!\!\nabla f_m(\xx({k-\tau_{m,k}}))}_{ \ee_{d}(k)}\notag 
		\notag\\&+\underbrace{\frac{1}{|\cM_k|}\sum_{m\in \cM_k}\nabla f_m(\xx({k-\tau_{m,k}}))-\hat{\gg}(k)}_{ \ee_{c}(k)}\label{error_term}
	\end{align}
	
	From \eqref{error_term}, we see that errors $\ee_c(k)$, $\ee_d(k)$ and $\ee_a(k)$ collectively compromise the  model update process. The overall error $\ee(k)=0$ when the conditions $\tau_{m,k}=0$, $\cM_k=\cM$, and $\gg(k)=\hat{\gg}(k)$ hold simultaneously.
	We are now ready to  analyze the impacts of the errors $\{\ee_c(k),\ee_d(k),\ee_a(k)\},\forall k$  on  the learning performance.
	
	\begin{proposition}\label{pro1}
		Under Assumptions \ref{as1}-\ref{as5}, the communication MSE is given by
	\eqref{com_error}, the device selection MSE is bounded by \eqref{de_sel_error}, and the model asynchrony MSE is bounded by \eqref{asy_error}, i.e., 
		\begin{align}
			\E\norm{\ee_c(k)}^2= & \!\sum_{c=1}^{C}\E \bigg[\bigg|\!\sum_{m\in \cM_k}\!\!\! \Delta_m(k) r_m^{(k)}(c) - \zeta(k)n^{(k)}(c)\bigg|^2\bigg],\label{com_error}\\
	\E\norm{\ee_d(k)}^2
	\leq& 8\Theta(k)\big[{(M-|\cM_k|)}/{M}\big]^2,\label{de_sel_error}\\
		\E\norm{\ee_a(k)}^2\leq& 6S^2\Theta(k) + 6T\E\sum_{j=k-S}^{k-1}\norm{\ee_c(j)}^2 + 2\delta^2,\label{asy_error}
\end{align}
where $	\Delta_m(k) = \frac{\sqrt{v_m{(k)}}}{|\cM_k|} - \zeta(k)|h_m(k)|b_m(k)$, $	\Theta(k)= \delta^2 + \alpha_1 + \alpha_2\E\left[\norm{\nabla F(\xx(k))}^2\right]$.
	\end{proposition}
	\begin{proof}
		Please refer to Appendix \ref{app_a}.
	\end{proof} 
	From Proposition \ref{pro1}, we find that the  communication MSE $\E \norm{\ee_c(j)}^2$ from the previous rounds also have impact on the model asynchrony MSE $\E\norm{\ee_a(k)}^2$.  
	To obtain a tractable expression of $\E \norm{\ee_c(j)}^2$, we determine it by leveraging the correlation of local gradients.
	In this paper,  we assume that the local gradients $\{\tilde{\gg}_m(k)\mid m \in \cM_k\}$ are  independent, and the local gradient elements $\{\tilde{g}_m^{(k)}(d)\mid m \in \cM_k, \forall d\}$ are  independent and identically distributed (i.i.d.). That is,
	\begin{align}
		&\E \left[ \tilde{\gg}_i(k)\left(\tilde{\gg}_j(k)\right)^\T  \right]=\mathbf{0},~ i\neq j \in \cM_k,\label{corre_1} \\&\E\left[\tilde{\gg}_m(k)\left(\tilde{\gg}_m(k)\right)^\T\right]=\mathbf{I}, ~m \in \cM_k\label{corre}.
	\end{align}
	We thus have the following corollary.
	\begin{corollary}\label{coro1}
		In the UAV-AFL system, under Assumptions \ref{as1}-\ref{as5}, together with the correlation assumption \eqref{corre_1} and \eqref{corre},  the communication MSE $\E\norm{\ee_{c}(k)}^2$ is bounded by
		\begin{align}
			&\E\norm{\ee_{c}(k)}^2\leq
		\left(\delta^2+\alpha_1+\alpha_2\norm{\nabla F(\xx(k))}^2\right)\notag\\&\quad\quad\bigg[|\cM_k|-\frac{(\sum_{m\in \cM_k}|h_m(k)|b_m(k))^2}{\left(\sum_{m\in \cM_k}|h_m(k)|^2b_m(k)^2+\delta^2_n/2\right)}\bigg]\label{c_mse}
		\end{align}
	\end{corollary}
	\begin{proof}
		Please refer to Appendix \ref{app_b}.
	\end{proof}

	Based on   Corollary \ref{coro1} and Proposition \ref{pro1}, we have the following Proposition.
	\begin{proposition}\label{pro_2}
		With Assumptions \ref{as1}-\ref{as5}, $\lambda=1/L$ and Corollary \ref{coro1},  $\E\left[F(\xx({K+1}))-F(\xx^\star)\right]$ is bounded by
	\begin{align}
		&\E\left[F(\xx({K+1}))-F(\xx^\star)\right]\leq \notag\\
		&\quad \prod_{k=1}^{K}\Bigg(1-\frac{\mu}{L}\Big[1-g(\qq_k,\aa_k,\bb_k)\alpha_2\Big]\Bigg)\E\left[F(\xx(1))-F(\xx^\star)\right] \notag\\
		&\quad +\frac{1}{2L}\sum_{k=1}^{K}\Bigg[\Big(g(\qq_k,\aa_k,\bb_k)(\delta^2+\alpha_1)+6\delta^2\Big)\notag\\
		&\quad \times \prod_{k'=k+1}^{K}\Bigg(1-\frac{\mu}{L}\Big[1-g(\qq_{k'},\aa_{k'},\bb_{k'})\alpha_2\Big]\Bigg)\Bigg], \label{pro2}
	\end{align}
		 where $g(\qq_k,\aa_k,\bb_k)= 18S^2+24(\frac{M-\sum_{m=1}^{M}a_m(k)}{M})^2+18S\sum_{j=k-S}^{k}(\sum_{m=1}^{M}a_m(j)-\frac{(\sum_{m=1}^{M}a_m(j)\sqrt{g_0}\norm{\qq(j)-\ww_m}^{-1}b_m(j))^2}{\sum_{m=1}^{M}a_m(j){g_0\norm{\qq(j)-\ww_m}^{-2}}b_m(j)^2+\delta^2_n/2})$.
	\end{proposition}
	\begin{proof}
		Please refer to Appendix \ref{app_c}.
	\end{proof}
	We see from Proposition \ref{pro_2} that  the overall convergence  $\E\left[F(\xx({K+1}))-F(\xx^\star)\right]$ is bounded by a quantity depending on UAV trajectory $\mQ$, device selection $\mA$,  and device transmit amplitude gain $\mB$.

	\section{System Optimization}\label{sys_opt}
	In this section, to improve the learning efficiency of the UAV-AFL system, we conduct a systematic design guided by the analysis in Section \ref{convergence}. Specifically, we establish the objective function based on  Proposition \ref{pro_2} and jointly optimize $\{\mQ,\mA,\mB\}$ to maximize the learning performance.
	\subsection{Problem Formulation}
	 From Proposition \ref{pro_2},   the condition  $1-\frac{\mu}{L}\left[1-g(i)\alpha_2\right]\leq 1$ must be held to guarantee convergence. Under this condition, as  $K \to \infty$,  $\E\left[F(\xx({K+1}))-F(\xx^\star)\right]$ is dominated by the second term on  the right hand side of \eqref{pro2}, i.e.,
	\begin{align}
		&\lim_{K\to \infty}\E\left[F(\xx({K+1}))-F(\xx^\star)\right]\leq \notag \\
		&\sum_{k=1}^{K}\frac{1}{2L}\left[g(\qq_k,\aa_k,\bb_k)(\delta^2+\alpha_1)+6\delta^2\right] \notag \\
		&\times \prod_{k'=k+1}^{K}\left(1-\frac{\mu}{L}\left[1-g(\qq_{k'},\aa_{k'},\bb_{k'})\alpha_2\right]\right). \label{asyp}
	\end{align}
	
	With the objective in \eqref{asyp}, we formulate the UAV-AFL   communication-learning design problem over  $\{\mQ,\mA,\mB\}$  as
	\begin{subequations}\label{eq28}
		\begin{align}
			\min_{\mQ,\mA,\mB}& \quad \sum_{k=1}^{K}\frac{1}{2L}\left[g(\qq_k,\aa_k,\bb_k)(\delta^2+\alpha_1)+6\delta^2\right]\notag\\&\times\prod_{k'=k+1}^{K}\left(1-\frac{\mu}{L}\left[1-g(\qq_k',\aa_k',\bb_k')\alpha_2\right]\right)\label{obj_1}\\
			\text{s.t.}
			&~~\,\sum_{k}^{k+S-1}a_m(k)\geq 1,\forall m,k\leq K-S+1,\label{staleness}\\
			&~~\,\sum_{k}^{k+c_m}a_m(k)\leq 1, \forall m, k \leq K-c_m, \label{active}\\
			&\quad
			a_m(k)\in \{0,1\}, \forall m,k\label{binary}\\ 
			&\quad {b_m(k)^2} \leq {P_0}/{2},\forall m,\eqref{mech_1}-\eqref{mech_3}\label{con_last}.
		\end{align}
	\end{subequations}
	where  \eqref{binary} denotes the binary constraints for the device selection variable $\mA$; \eqref{staleness} is the staleness constraint from Assumption \ref{as3}, which is introduced to mitigate the impact of stale models on the learning process; \eqref{active} represents that the devices cannot be selected during the local computing process; ${b_m(k)^2} \leq {P_0}/{2}$ is the device transmit power constraint;  \eqref{mech_1}-\eqref{mech_3} are the mechanical constraints of the UAV-PS.

	\begin{remark}
		Problem \eqref{eq28} is a mixed-integer programming problem. To address this, we use geometric programming (GP) to transform the objective function \eqref{obj_1} into a tractable convex form. Then, an efficient two-layer penalty based  algorithm using alternative optimization (AO) and successive convex approximation (SCA) is developed to solve this problem, as discussed in what follows.
	\end{remark}
	\subsection{Problem Transform via Geometric Programming } 
	To deal with the non-convex objective function, we introduce auxiliary variables $\pp\in \mathbb{R}^{2K}$ and $\yy\in \mathbb{R}^{2K}$ to recast problem  \eqref{eq28} as 
	\begin{subequations}\label{recastpro1}
		\begin{align}
			& \min_{\mQ,\mA,\mB,\pp,\yy} \quad \log \left(\sum_{k=1}^{K}e^{\cc_k^\T \yy}\right)\\
			\text{s.t.}\quad&
			\frac{\delta^2+\alpha_1}{2L}g(\qq_k,\aa_k,\bb_k)+3\frac{\delta^2}{L}\leq 	p_k, \forall k,\label{33b}\\
			&1\!\!-\!\!\frac{\mu}{L}\!\!+\!\!\frac{\mu}{L}g(\qq_k,\aa_k,\bb_k)\alpha_2 \!\leq\! p_{K+k},\forall k,\label{33c}\\
			&p_i\leq e^{y_i}, \forall i\label{33d}
			\\& \eqref{binary}-\eqref{con_last},
		\end{align}
	\end{subequations}
	where $\cc_k\in \mathbb{R}^{2K}$ with  $c_{k}^k=c^{K+k+1}_k=c^{K+k+2}_k=\cdots=c^{2K}_k=1$ and  other elements of $\cc_k$ are $0$.
	\begin{proposition}\label{pro_3}
		Problem \eqref{recastpro1} is an equivalent  form of \eqref{eq28}.
	\end{proposition}
	\begin{proof}
		The  auxiliary variable $\pp$ satisfying \eqref{33b} and \eqref{33c} is used to recast  \eqref{eq28} as
		\begin{subequations}\label{recast_proof}
			\begin{align}
				\min_{\mQ,\mA,\mB,\pp}& \quad \sum_{k=1}^{K}p_k\prod_{k'=k+1}^{K}p_{K+k'}\label{obj_p}\\
				\text{s.t.}
				&\quad \eqref{33b},\eqref{33c},\eqref{binary}-\eqref{con_last}.
			\end{align}
		\end{subequations}
	
		Note that the objective \eqref{obj_p} monotonically increases with respect to $\pp$. Hence, the optimal $\pp$ must be obtained  at the lower boundary of the feasible region  of problem \eqref{recast_proof}. Consequently,  at the optimum of \eqref{recast_proof}, constraints  \eqref{33b} and \eqref{33c} are satisfied with equality. This means the equivalence of problems \eqref{recast_proof} and \eqref{eq28}. Moreover, the equivalence between  \eqref{recast_proof} and \eqref{recastpro1} can be proved similarly by considering  auxiliary variable $\yy$. The proof is thus complete.
	\end{proof}
	
	\subsection{Penalty Based Binary Constraint} 
	The conventional method \cite{9891794},\cite{wu2018joint} for tackling the binary device selection constraint \eqref{binary} is to convert the binary elements of $\mA$ to continuous-valued variables in optimization and then rounding the optimized  $\mA$. However,  the obtained $\mA$ by this way may not satisfy the staleness constraint \eqref{staleness} and the computing time constraint \eqref{active}. In this paper, a penalty based method is used to handle this constraint. We initially introduce auxiliary variables to transform \eqref{staleness} into a series of equality constraints. Subsequently, to ensure the satisfaction  of these equality constraints, we incorporate them into the objective function in the form of penalty terms.   Specifically,  we introduce 
	$\bar{\mA}=\{\bar{a}_m(k),\forall m,k\}\in \mathbb{R}^{M\times K}$ and rewrite \eqref{binary} as 
	\begin{align}
		a_m(k)(1-\bar{a}_m(k))=0,~a_m(k)=\bar{a}_m(k)\label{equa_bin},\forall m,k.
	\end{align}
	Clearly, the elements of $\mA$ must be either $0$ or $1$ to comply with the constraints \eqref{equa_bin}. We then introduce penalty terms  in the form of \eqref{equa_bin} into the objective function  \cite{bertsekas1997nonlinear}, yielding the following  formulation:
	\begin{subequations}\label{recastpro2}
		\begin{align}
			\min_{\mQ,\mA,\bar{\mA},\mB,\pp,\yy}& \quad \log \left(\sum_{k=1}^{K}e^{\cc_k^\T \yy}\right)+\frac{1}{\eta}\sum_{m=1}^{M}P(\mA,\bar{\mA})\\
			\text{s.t.}
			&\quad \eqref{33b},\eqref{33c},\eqref{33d},\eqref{staleness}-\eqref{con_last},
\		\end{align}
	\end{subequations}
	where $P(\mA,\bar{\mA})=[a_m(k)(1-\bar{a}_m(k))]^2+[a_m(k)-\bar{a}_m(k)]^2$, and $\eta>0$ is employed to penalize any deviation from the equality constraint \eqref{equa_bin}. It can be easily verified that the binary  constraint \eqref{binary} must be satisfied  as $\frac{1}{\eta} \to \infty$. 
%
	\subsection{Inner and Outer Layer Iteration}
	
	In this subsection, we propose a two-layer iterative algorithm. Specifically, in the inner layer, we decompose \eqref{recastpro2} into three sub-problems, wherein the  variables $\bar\mA$, $\mA$ and $\{\mQ,\mB\}$ are optimized separately. In the outer layer, we  gradually decrease the coefficient $\eta$. The details are shown as follows.
	
	\subsubsection{Slack Variable Optimization}
	We first optimize slack variable $\bar{\mA}$ with given $\{\mA,\mQ,\mB\}$. By dropping the irrelevant terms,  
	$\bar{\mA}$ appears exclusively in the objective function as a convex form, and the optimal value of  $\bar{\mA}$ is
	\begin{align}
		\bar{a}_m^\star (k)=\frac{a_m^2(k)+a_m(k)}{1+a_m^2(k)}\label{up_bar}
	\end{align}
	
	\subsubsection{Device Selection Optimization} We now optimize the device selection variable $\mA$. With given $\{\bar{\mA},\mQ,\mB\}$, problem \eqref{recastpro2} can be reformulated as
	\begin{subequations}\label{recastpro3}
		\begin{align}
			\min_{\mA,\pp,\yy}& \quad \log \left(\sum_{k=1}^{K}e^{\cc_k^\T \yy}\right)+\frac{1}{\eta}\sum_{m=1}^{M}P(\mA,\bar{\mA})\label{pro3_obj}\\
			\text{s.t.}
			&\quad g(\aa_k)\leq \Lambda(k),\forall k\label{g_leq_1}\\
			&\quad \eqref{staleness},\eqref{active},\eqref{33d},\label{37c}
		\end{align}
	\end{subequations}
	where $\Lambda(k)=\min\left( \frac{2Lp_k-6\delta^2}{\delta^2+\alpha_1},\frac{Lp_{K+k}+\mu-L}{\mu\alpha_2} \right)$.
	The problem \eqref{recastpro3} is neither convex or quasi-convex because of the non-convex constraints \eqref{g_leq_1} and \eqref{33d}. Here, we introduce  auxiliary variables $\{\ff,\dd\}$  and use the successive convex approximation (SCA) technique to transform  \eqref{recastpro3} to a convex form as
	\begin{subequations}\label{recastpro4}
		\begin{align}
			&\min_{\mA,\pp,\yy,\ff,\dd} \quad \log \left(\sum_{k=1}^{K}e^{\cc_k^\T \yy}\right)+\frac{1}{\eta}\sum_{m=1}^{M}P(\mA,\bar{\mA})\label{pro4_obj}\\
			\text{s.t.}
			&\quad 		f_k\leq \Lambda(k),\forall k,\\
			&\quad   p_i\leq e^{y_i^{t}}+e^{y_i^{t}}(y_i-y_i^{t}), \forall i,\label{exp_y}\\
			&\quad f_k \geq 18S^2 \!\!+\!\! \frac{24}{M^2}\Phi(k)^2 \!\!+\!\! 18S\Psi(k),\forall k, \label{40_cons}	\\
			&\quad 	d_k \leq R(k) \!\!+\!\! (\aa_k\!-\!\aa^t(k))^\T(\Upsilon(k)+\Xi(k)), \forall k,\!\! \label{43_cpms}\\
			&\quad \eqref{staleness},\eqref{active},
		\end{align}
	\end{subequations}
where $\{\mA^t,\yy^t\}$  are the optimized results of $\{\mA,\yy\}$ obtained at the previous inner iteration,  $\aa^t(k)$ is the $k$-th column of $\mA^t$,  $	\Psi(k) =\sum_{j=k-S}^{k}(\sum_{m=1}^{M}a_m(j)-d(j))$, $\Phi(k) = M-\sum_{m=1}^{M}a_m(k)$, $\vv(k)=[|h_1(k)|b_1(k),\cdots,|h_M(k)|b_M(k)]$,  $R(k) =\frac{\aa^t(k)^\T\vv(k)\vv(k)^\T\aa^t(k)}{\aa^t(k)^\T\vv(k)\circ\vv(k)+\delta^2_n/2}$, $\Omega(k) = \frac{2\vv(k)\vv(k)^\T\aa^t(k)}{(\aa^t(k)^\T\vv(k)\circ\vv(k)+\delta^2_n/2)^2}$, $\Upsilon(k)= \Omega(k) \cdot (\aa^t(k)^\T\vv(k)\circ\vv(k)+\delta^2_n/2)$, $	\Xi(k) = -\frac{\aa^t(k)^\T\vv(k)\aa^t(k)^\T\vv(k)\vv(k)\circ\vv(k)}{(\aa^t(k)^\T\vv(k)\circ\vv(k)+\delta^2_n/2)^2}$.

	\begin{proposition}\label{propo4}
		By solving problem \eqref{recastpro4}, we obtain a locally optimal solution of  \eqref{recastpro3} satisfying the Karush-Kuhn-Tucker (KKT) condition, which serves as an upper bound of \eqref{pro3_obj}.
	\end{proposition}
	\begin{proof}
		Please refer to Appendix \ref{app_d}.
	\end{proof}
	\subsubsection{Trajectory and Transmit Power Optimization}  
	Similar to the discussions of \eqref{recastpro3},  with fixed $\{\mA,\bar{\mA}\}$,  \eqref{recastpro2} can be solved by using SCA  to obtain an update of $\{\mQ,\mB\}$. Due to page limitations, we omit the details here. 

	\subsubsection{Outer layer update}In the outer layer, the penalty coefficient $\eta$ is iteratively updated by  $\eta=s\eta$, where $s$ represents a scaling factor with range $0<s<1$. 
	\subsection{Overall Algorithm and Complexity Analysis}
	We summarize the proposed algorithm in Algorithm \ref{alg}.
	\begin{algorithm}[h]
		\caption{Two-layer Iterative Algorithm}
		\label{alg} 
		\begin{algorithmic}[1] 
			\STATE {\textbf{Initialize:}} The initial value $\bar{\mA}^0,\mQ^0,\mA^0$ and $\mB^0$,  iteration index $t=1$,  predefined  precision $\epsilon_1$ and $\epsilon_2$.
			\REPEAT 
			\REPEAT 
			\STATE With given $\{\mQ^t,\mA^t,\mB^t\}$, compute $\bar{\mA}^{t+1}$ based on \eqref{up_bar}.
			\STATE With given $\{\bar{\mA}^{t},\mQ^t,\mB^t\}$, compute $\mA^{t+1}$ by solving problem \eqref{recastpro4}.
			\STATE With given $\{\bar{\mA}^{t},\mA^t\}$, compute $\{\mQ^{t+1},\mB^{t+1}\}$ by solving the convexifed problem. 
			\STATE Calculate  the value $L^{t+1}$ of objective function in \eqref{recastpro2}.
			\STATE $t = t + 1$.
			\UNTIL $\left|(L^{t} - L^{t-1})/L^{t-1}\right| \le \epsilon_1$
			\STATE $\eta=s\eta$.
			\UNTIL the violation of \eqref{equa_bin} falls below the threshold $\epsilon_2$.
			\ENSURE $\{\mQ,\mA,\mB\}$.
		\end{algorithmic}
	\end{algorithm}

We now analyze the complexity of  Algorithm \ref{alg}.
	The computational complexity of  Algorithm \ref{alg} is dominated by steps 5 and 6. We employ the convex optimization tool CVX \cite{boyd2004convex}, which utilizes the standard interior-point method, to address the problems posed in these steps. Consequently, the complexity involved in steps 5 and 6 is $\mathcal{O}((MK)^{3.5})$.
	Therefore, the overall complexity of Algorithm  \ref{alg} is $\mathcal{O}(I_{outer}I_{inner}(MK)^{3.5})$, where $I_{outer}$ and $I_{inner}$ denote the numbers of iterations  to achieve the  precision in the inner and outer layers, respectively.
	
	\section{Numerical Result}
	
	In this section, we present numerical results to evaluate the performance of the proposed UAV-AFL scheme. These results offer valuable insights for the UAV-AFL implementation and demonstrate the superior performance of the proposed scheme.
	\subsection{Simulation Setup}
	Unless otherwise specified, the system parameters are configured as follows. We consider an  area of 1 km × 1 km with a total of $M=20$  edge devices. The maximum  flight speed $v_{\max}$ (or acceleration $a_{\max}$) of the UAV-PS is $50$ m/s (or $15$ m/$\text{s}^2$), and the fixed flying altitude $H=100$ m. We set the initial and final dispatch  point of UAV-PS as $\qq_0=\qq_F(0,0,100)$; the reference channel  gain $g_0=-60$ dB; the noise power $\sigma_n^2=-80$ dBm;  the maximum transmit power of devices  $P_0=0.1$ W; the staleness bound $S=50$, the number of overall time slots $K=10000$; the duration of each time slot $\delta_t=1$ s; the  threshold $\epsilon_1$  (or $\epsilon_2$) is $10^{-2}$ (or $1$).
	Besides, we set  $\mu=0.2$, $L=10$, $\delta=20$, $\alpha_1=100$ and $\alpha_2=0.1$. To provide a clear exhibition  of the UAV-PS trajectory, we chose $K=250$  to  obtain the optimization results, and use this as a basic cycle to  conduct the UAV-AFL system training periodically for $10000$ time slots. 
	
	For the learning configuration, we conduct image classification task using the MNIST dataset. We train a convolutional neural network with two $5\times5$ convolutional layers (each followed by $2\times2$ max pooling), along with a batch normalization layer, a fully connected layer consisting of 50 units, a ReLu activation layer, and finally a softmax output layer. The loss function employed is the cross-entropy loss. The local data of edge devices are i.i.d. and equally drawn from a pool of $5\times 10^4$ images while the test data set contains $10^4$ images.
	
	\subsection{Comparisons Based on Different Staleness Bound}
	In this subsection, we investigate the impact of different staleness bounds $S$ on the performance of the UAV-AFL system. 
	We conduct simulations  in three  scenarios.  In scenario 1, the devices are randomly located. In scenario $2$ and $3$,   the devices are located within different clusters.
	We use the metric $\sigma(\ww)\triangleq\sum_{m=1}^{M}\norm{\ww_m-\overline{\ww}_m}$ to roughly measure the standard deviation of devices' geographical distributions. 
	 We introduce the  communication  normalized mean square error (NMSE) metric, i.e., $10\log_{10}\left(\E\left[\norm{\hat{	\gg}(k)-	\gg(k)}^2\right]/\norm{\gg(k)}^2\right)$  to quantify the communication quality of time slot $k$. Table~\ref{tab:NMSE} shows the average of the communication NMSE under $K$ time slots. All the presented results are averaged over  $50$ Monte Carlo trials. 
		\begin{figure*}[h]
	\setlength{\belowcaptionskip}{-20pt}
	\setlength{\abovecaptionskip}{-0.1cm} 
	\centering
	\subfigure[Trajectory of Scenario 1.]{
		\includegraphics[width=5.1cm, height=1.8in]{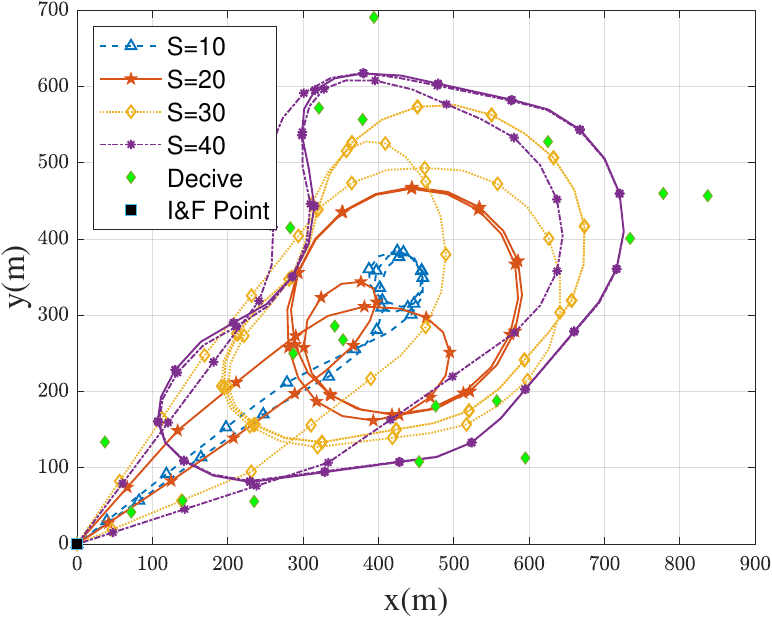}}
	\label{fig1a}
	\subfigure[Trajectory of Scenario 2.]{
		\includegraphics[width=5.1cm, height=1.8in]{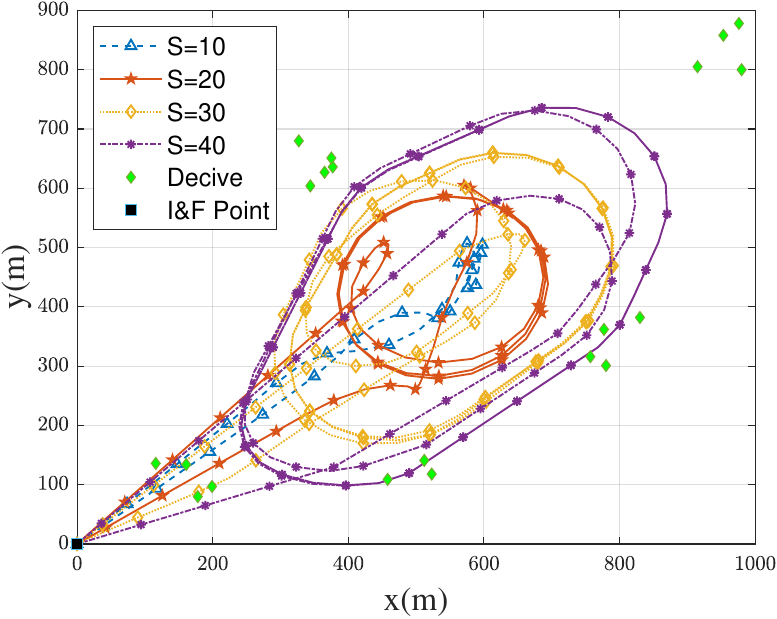}}
	\label{fig2a}
	\subfigure[Trajectory of Scenario 3.]{
		\includegraphics[width=5.1cm, height=1.8in]{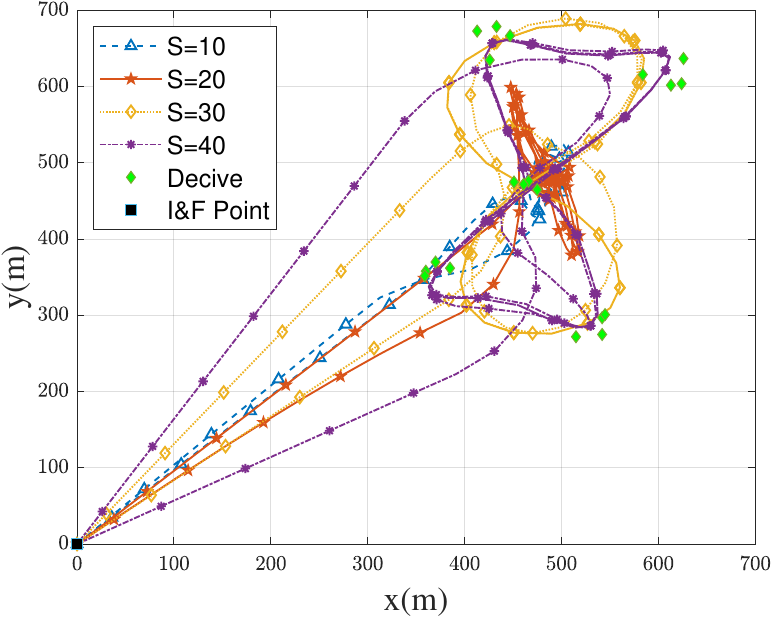}}
	\label{fig3a}
	\caption{Trajectory of different staleness bounds in three scenarios.}
	\vspace{-3mm}
	\label{fig_traj_stale}
\end{figure*}
\begin{figure*}[h]
	\setlength{\belowcaptionskip}{-20pt}
	\setlength{\abovecaptionskip}{-0.1cm} 
	\centering
	\subfigure[Performance of Scenario 1.]{
		\includegraphics[width=5.1cm, height=1.8in]{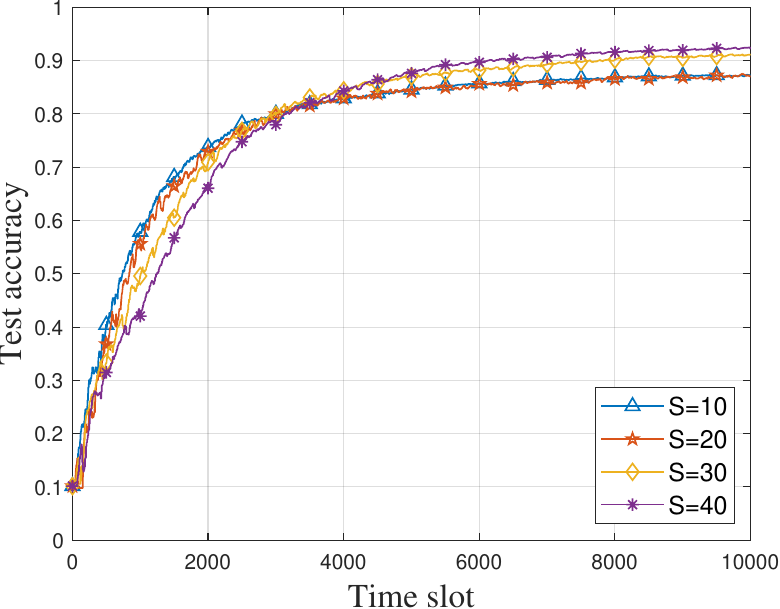}}
	\label{fig1b}
	\subfigure[Performance of Scenario 2.]{
		\includegraphics[width=5.1cm, height=1.8in]{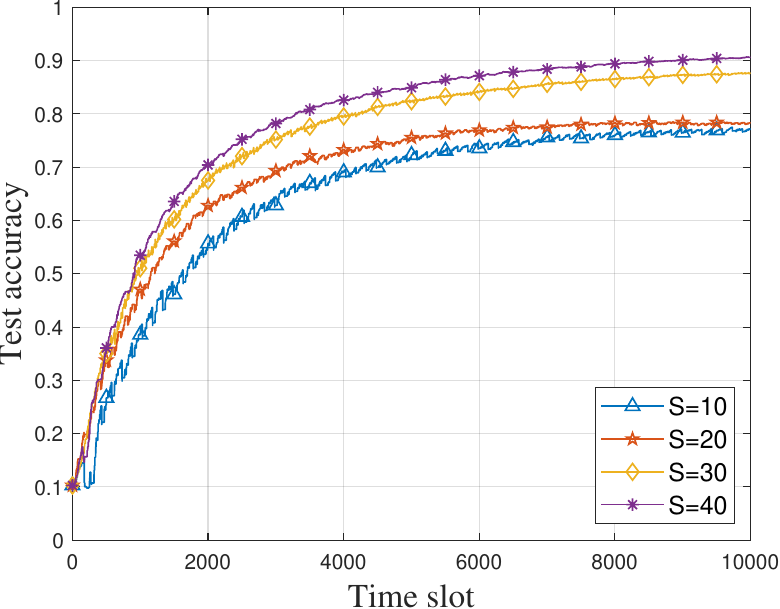}}
	\label{fig2b}
	\subfigure[Performance of Scenario 3.]{
		\includegraphics[width=5.1cm, height=1.8in]{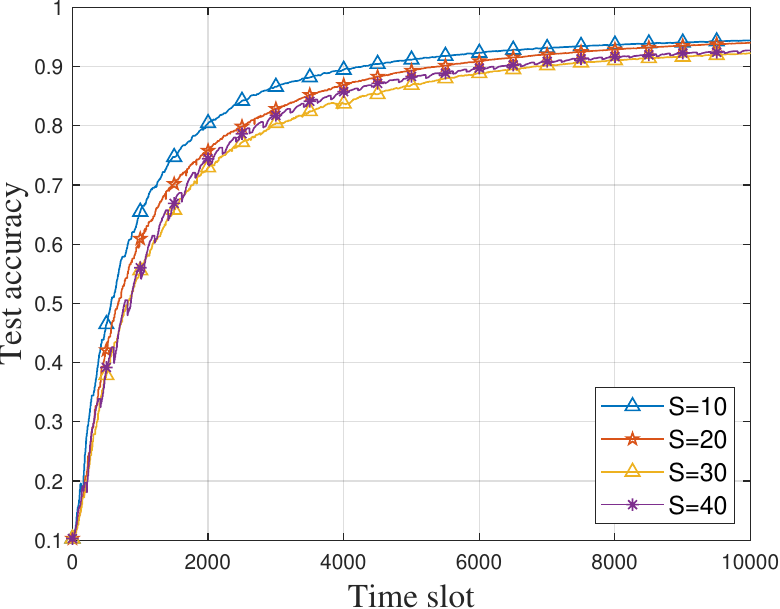}}
	\label{fig3b}
	\caption{Performance of different staleness bounds in three scenarios.}
	\vspace{-3mm}
	\label{fig_acc_stale}
\end{figure*}

	
	From Fig.~\ref{fig_traj_stale}, it is evident that a larger staleness bound  provides  UAV-PS more  time  to approach devices for  gradient collection. This  results in more expansive flight trajectories and  reduced communication NMSE, as shown in Table~\ref{tab:NMSE}.  Conversely, with a more stringent staleness constraint, the UAV-PS is compelled to complete a service cycle more rapidly, resulting in  narrower trajectories.
	In the simulations, we note the UAV-PS  stops communicating with devices under suboptimal channel conditions, especially when there is no proximal device. This behavior underscores the efficacy of our proposed algorithm in balancing device selection with communication quality, where the UAV-PS only aggregates the gradients  when communication quality is favorable.
	
	In Fig.~\ref{fig_acc_stale}, the learning performance varies distinctly across different staleness upper bounds $S$. 
	For devices located in close proximity (scenario 3, with $\sigma(\ww)=153.36$ m), the UAV-PS can have favorable channel quality without the necessity of nearing each device. Under such circumstances, smaller $S$  values exhibit higher learning efficiency, attributable to shorter service cycles and faster model update rate.
	For scenario 2 where devices' location are more dispersed ($\sigma(\ww)=384.53$ m), communication quality becomes the main challenge. Here, a larger  $S$ (i.e., $S=30, 40$) allows the UAV-PS the enough  time to approach devices and enhance the quality of  transmission, leading to improved performance. Lastly,
	scenario 1, with $\sigma(\ww)=269.75$ m, represents an intermediate case between scenarios 2 and 3. In this scenario, at the early stages of training ($k<4000$), the learning accuracy is predominantly impacted by the model update rate, making schemes with a smaller $S$ superior. However, as training progresses ($k>4000$), the adverse effects of communication errors become  pronounced (as the model  stabilizes). In this phase, schemes with a larger $S$ exhibit better test accuracy.
		\begin{table}[h]
		\caption{Average Communication NMSE}
		\centering
		\small
		\begin{tabular}{l|cccccc}
			\hline
			\multicolumn{1}{c|}{\multirow{3}{*}}              & \multicolumn{6}{c}{NMSE ({dB})}                                            \\ \cline{2-7} 
			\multicolumn{1}{c|}{}                                                  & \multicolumn{2}{c|}{Scenario 1 }    & \multicolumn{2}{c|}{Scenario 2} & \multicolumn{2}{c}{Scenario 3}\\ \cline{2-7} 	\hline
			S=10  & \multicolumn{2}{c|}{${-5.3743}$} 
			& \multicolumn{2}{c|}{${-1.4945}$} 
			&\multicolumn{2}{c}{${-12.0915}$}     \\
			S=20 & \multicolumn{2}{c|}{${-7.3219}$} 
			& \multicolumn{2}{c|}{${-3.0754}$} 
			&\multicolumn{2}{c}{${-12.7055}$}         \\
			S=30               & \multicolumn{2}{c|}{${-8.3004}$} 
			& \multicolumn{2}{c|}{${-6.8537}$} 
			&\multicolumn{2}{c}{${-14.3379}$}     \\
			S=40               & \multicolumn{2}{c|}{${-10.4780}$} 
			& \multicolumn{2}{c|}{${-9.4066}$} 
			&\multicolumn{2}{c}{${-16.9662}$}     \\
			\hline
		\end{tabular}
		\label{tab:NMSE}
	\end{table}
	\subsection{Performance Under Heterogeneous  Computing Ability}
	In this subsection, we conduct experiments under the FL system where edge devices have heterogeneous computing abilities.  In the simulations, devices are grouped into five distinct clusters based on their locations, where the devices within  the same cluster are set  to have the same gradient computation time $c_m$ defined in Section \ref{sec3subA}. 

	In Fig.~\ref{fig_44}, we study the performance of UAV-AFL scheme  under two scenarios with different device distributions. In each scenario, both heterogeneous and homogeneous configurations are considered. For the heterogeneous setting, the  computation time  $c_m$ for devices within three  clusters (cyan diamond icons) is $5$ s  while the computation time for devices in other clusters (green pentagon icons) is $30$ s. For the homogeneous  setting, the computation time $c_m$ for each device is set to $15$ s.\footnote{By this setting, the total  time  for all devices to complete computation  are the same under both heterogeneous and homogeneous settings.}
	
	 We  see from Fig.~\ref{fig_44} (a) and (b) that the UAV-PS alters its flight strategy  in the two different settings significantly. In the homogeneous setting, the UAV-PS appears to  traverse and poll each device cluster in sequence. Conversely, in the heterogeneous setting,  the UAV-PS prioritizes devices with superior  computational abilities.  
	 For example, in scenario 1, devices with superior computational capabilities ($c_m=5$ s) are selected, on average,  $1.92$ times more frequently than  less capable devices ($c_m=30$ s).
	Hence, the proposed flight strategy can leverage the  gradient computation time to facilitate the model updates.
	
	In Fig.~\ref{fig_44} (c), we note that the heterogeneity adversely impacts the UAV-AFL learning performance.  This   can be understood from two perspectives. Firstly, the  heterogeneity  prompts  the UAV-PS to rely heavily  on devices with superior computational capabilities for model updates, which introduces a bias in the learned model.
	Secondly, heterogeneous settings often result in larger staleness degree compared to homogeneous settings. We quantify this by calculating the  maximum staleness (MS), defined as $\max{\{\tau_{m,k}\mid\forall m, \forall k \}}$, within our simulation. This can be considered as the  tightest staleness bound $S$ achieve by the practical UAV-AFL system.  For example, in Scenario (a), the MS values for heterogeneous and homogeneous configurations are $40$  and $28$, respectively.  As suggested by Proposition \ref{pro_2},  the learning performance is compromised by the large value of MS.
			\begin{figure*}[t]
		\centering
		\subfigure[Trajectories of scenario 1.]{
			\includegraphics[width=5.1cm, height=1.8in]{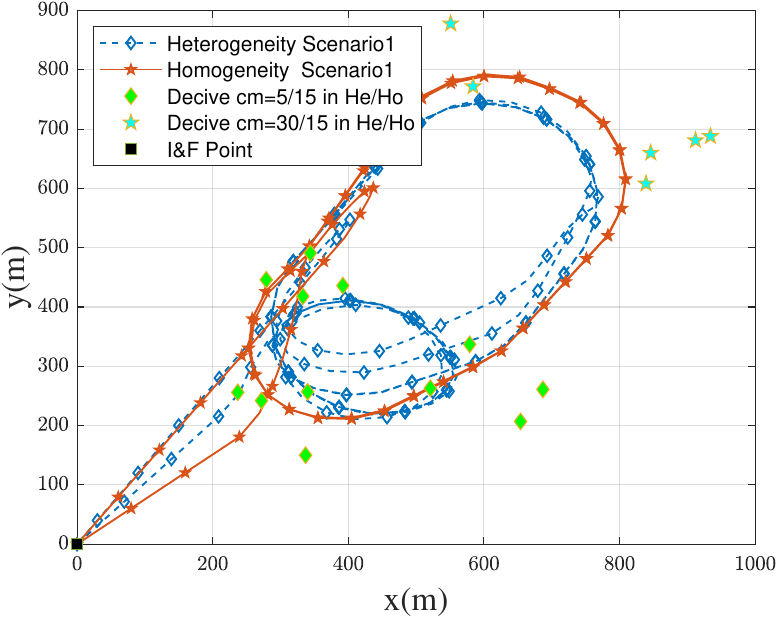}}
		\label{fig4a}
		\subfigure[Trajectories of scenario 2.]{
			\includegraphics[width=5.1cm, height=1.8in]{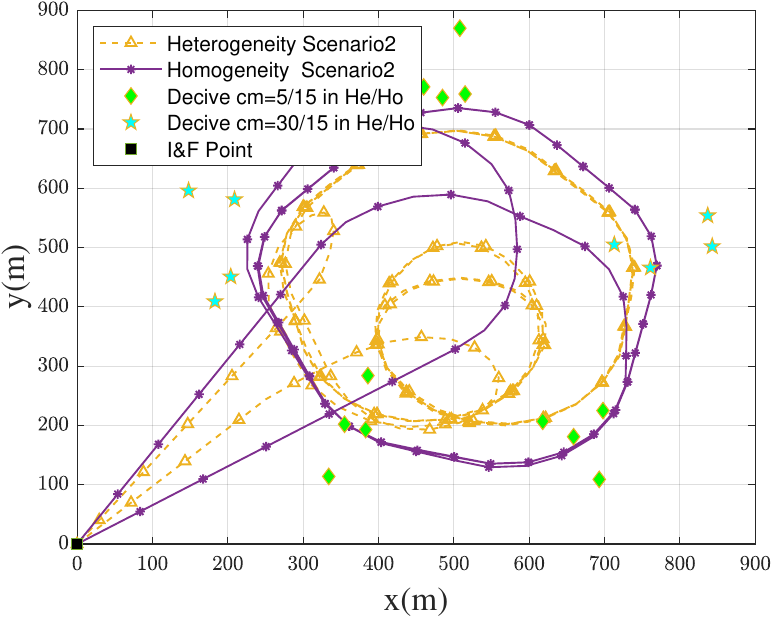}}
		\label{fig4b}
		\subfigure[Performance of scenario 1 and 2.]{
			\includegraphics[width=5.1cm, height=1.8in]{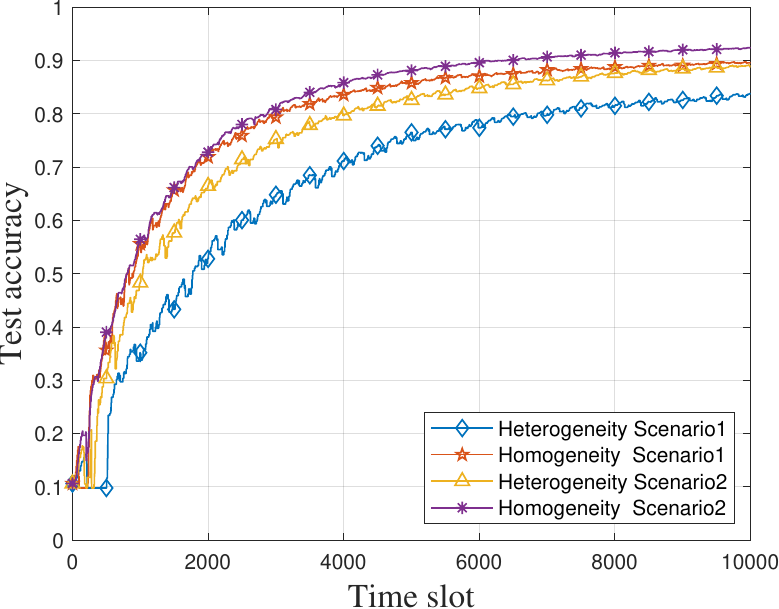}}
		\label{fig4c}
		\caption{System performance under heterogeneous and homogeneous  settings.}
		\vspace{-3mm}
		\label{fig_44}
	\end{figure*}
	\subsection{Performance Comparison With Benchmarks}
	In this subsection, we present a comparison of  the proposed scheme with several existing state-of-the-art solutions to demonstrate the superior UAV-AFL  performance.
	The benchmarks are as follows.
	\begin{itemize}
		\item \textbf{Error  free channel (Error free)}: We use the device schedule result i.e., $\mA$ obtained by the Algorithm \ref{alg} to update the global model without communication error, i.e.,  $\hat{\gg}(k)=\gg(k)$. 
		\item \textbf{Circular flight trajectory (CF)\cite{7888557}}: The UAV-PS flies with a circular trajectory \cite{7888557} centered at $(500,500,H)$ using the maximum velocity.  
		The  radius of trajectory $R$ is determined heuristically by the   average distance minimization problem, i.e., $\min_{R} \sum_{i=1}^{M} (R-\sqrt{x_m^2+y_m^2})^2$. Given the trajectory, the device selection and the transmit amplitude gain are optimized by Algorithm \ref{alg}.
		\item \textbf{Synchronous Fly-Hover-Fly (SFHF)\cite{amiri2020federated}}: The UAV-PS proceeds directly from its initial position to a optimized location. At this location, the UAV-PS serves as a central PS  providing synchronous FL update \cite{amiri2020federated} where all devices upload the gradients over-the-air. The location $\qq$ is determined by
		\begin{align}
			\min_{\qq,\bb}\quad 
			M-\frac{(\sum_{m\in \cM}{\sqrt{g_0\norm{\qq-\ww_m}^{-2}}b_m)^2}}{\sum_{m\in \cM}g_0\norm{\qq-\ww_m}^{-2}b_m^2+\delta^2_n/2}.\label{c_mse1}
		\end{align}
		\item  \textbf{Hierarchical gradient aggregation (HGA) \cite{zhong2022uav}}: The service strategy  in \cite{zhong2022uav} is simulated. Under this approach,   
		the UAV-PS is  programmed to  aggregate the gradients of devices located within a  threshold of $d_{\text{thr}}=250$ m from the UAV-PS over-the-air. After all the gradients are collected, the UAV-PS further combines the previous  aggregated gradients and updates the global model. 
		\item  \textbf{Fully asynchronous FL with error-free communication (FAFL) \cite{xie2019asynchronous}}: Following the fully asynchronous FL principle  \cite{xie2019asynchronous}, only a singular  edge device is permitted to upload its gradient during one time slot. In this paradigm, the UAV-PS is mandated to sequentially approach each device to ensure effective gradient transmission. 
		Given that each transmission is under a point-to-point communication link with low path loss, it is reasonable to assume error-free communication. The  technique for solving the traveling salesman problem (TSP) is employed to optimize the UAV-PS trajectory in this scheme.
	\end{itemize}
	\begin{figure}[h]
		\centering
		\subfigure[Trajectories]{
			\includegraphics[width=5.5cm, height=4.5cm]{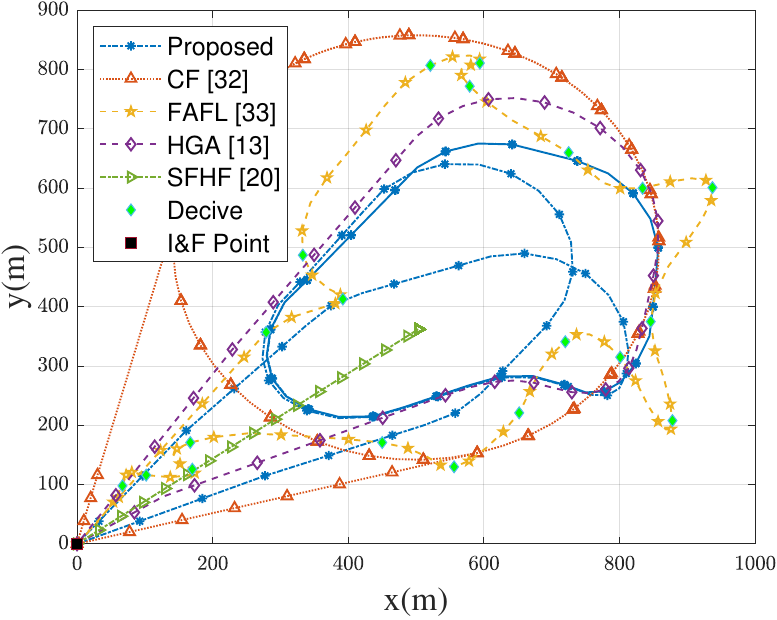}}
		\label{fig5a}
		\subfigure[Test accuracies]{
			\includegraphics[width=5.5cm, height=4.5cm]{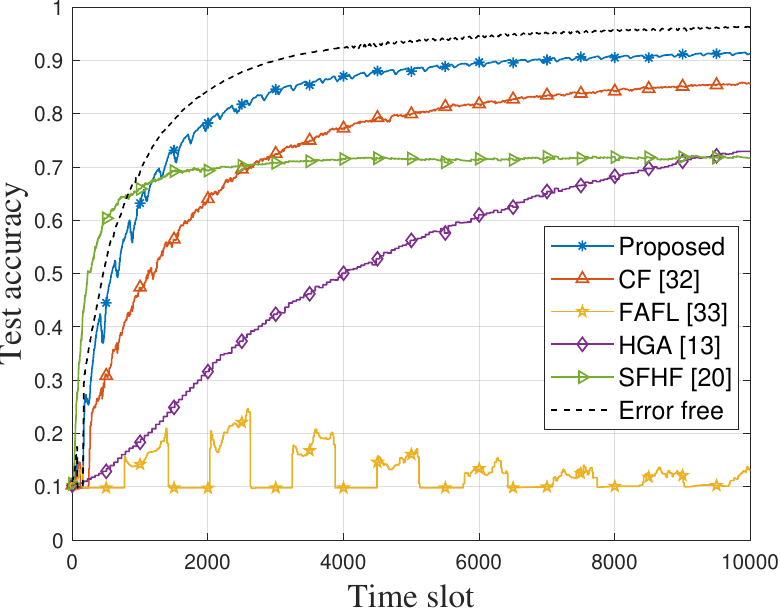}}
		\label{fig5b}
		\caption{Performance comparison of different schemes in scenario 1.}
		\label{fig55}
	\end{figure}
			\begin{figure}[h]
		\centering
		\subfigure[Trajectories]{
			\includegraphics[width=5.5cm, height=4.5cm]{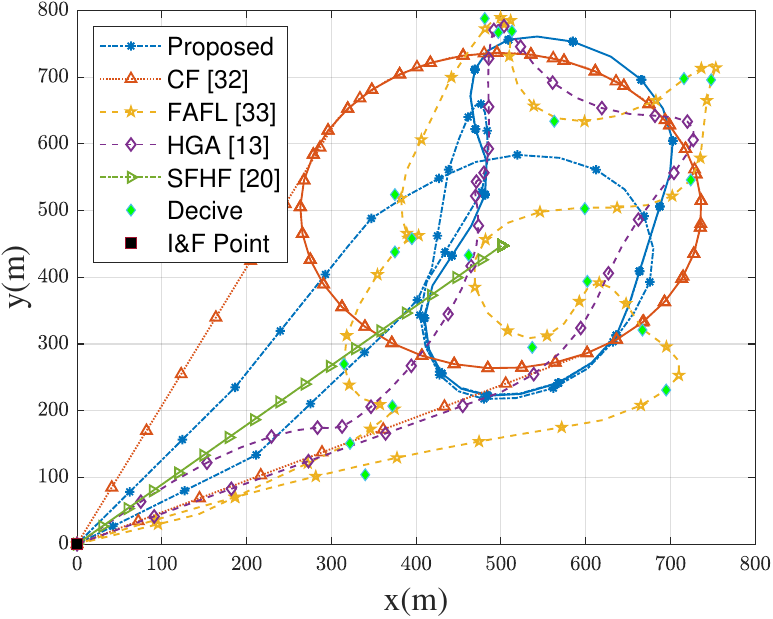}}
		\label{fig6a}
		\subfigure[Test accuracies]{
			\includegraphics[width=5.5cm, height=4.5cm]{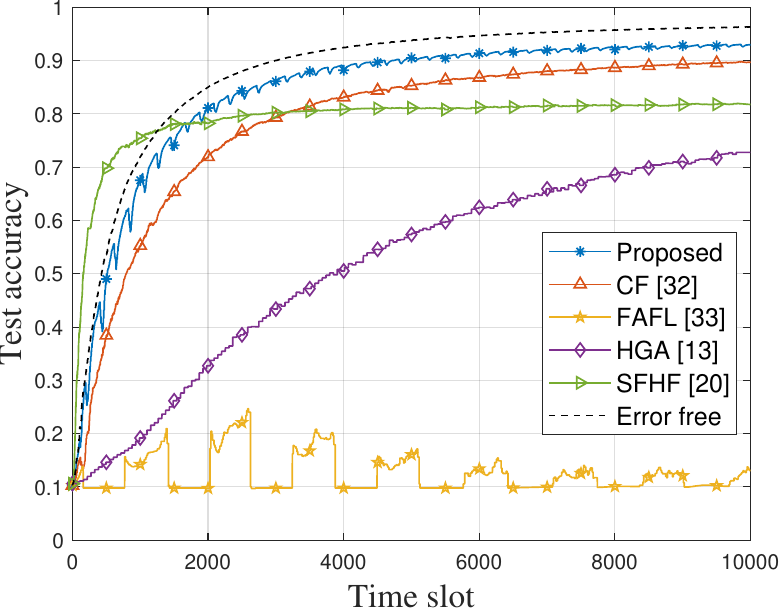}}
		\label{fig6b}
		\caption{Performance comparison of different schemes in scenario 2.}
		\label{fig66}
	\end{figure}
	
	In Fig.~\ref{fig55}, we depict the trajectory and the test accuracy of various schemes in a scenario where devices are dispersedly located ($\sigma(\ww)=340.82$ m).   
	It is evident that the proposed UAV-AFL scheme offers a  substantial performance enhancement  in comparison to the benchmarks, and is similar to the error-free case. 
	The disparities in learning accuracy across different schemes can be attributed to the characteristics of their service strategies.
	 Specifically,
	although the accuracy of SFHF sees a rapid rise at the beginning (due to the use of all devices' gradients for updates), it becomes stabilized at a less-than-optimal level (due to the restriction by the poor communication quality). 
Moreover,  because of the fully asynchronous nature of FAFL, this scheme suffers from a pronounced staleness degree ($\text{MS}=114$). Hence, its accuracy oscillates around $0.1$ throughout the training period.
	Furthermore, the learning accuracy  under the HGA method experiences a  slowly ascent throughout this period. This is because the hierarchical strategy  necessitates the waiting for all devices' gradient collection, which inevitably impair the learning efficiency.  Lastly, the inefficient performance of CF approach highlights the importance of trajectory optimization.

In Fig.~\ref{fig66}, we present the trajectory and the test accuracy in a scenario where the devices are densely located  ($\sigma(\ww)=213.31$ m). The results are similar with those depicted in Fig.~\ref{fig55}, which demonstrates the robust performance gain achieved by the proposed scheme.
	
	\section{Conclusions} \label{sec-conclusion}
	In this paper, we investigated the design of the  UAV-AFL system over edge networks. We proposed   a novel instant-update service strategy to improve the learning efficiency. We introduced  a staleness bound metric to limit the level of AFL model asynchrony and derived an overall convergence bound  by capturing the impact of model asynchrony, device selection and communication errors on the UAV-AFL learning performance. This   bound represents a systematic endeavor to evaluate  AFL performance under the consideration of both communication and learning perspectives. Subsequently, an optimization problem is then formulated to jointly determine UAV-PS trajectory, device selection and over-the-air transceiver design.  To address this, we proposed a two-layer iterative algorithm leveraging   geometric programming (GP) and successive convex approximation (SCA) for efficient problem-solving. Our numerical results  demonstrated the superior  effectiveness of the proposed UAV-AFL scheme when benchmarked against  state-of-the-art approaches.
	\appendices
	\section{Proof of Proposition~\ref{pro1}}\label{app_a}
	 For notation convenience, we denote $\nabla f_m(\xx({k-\tau_{m,k}})$ and  $\nabla F(\xx({k-\tau_{m,k}}))$ as $\nabla f_{m,k}$ and $\nabla F_{m,k}$, respectively. Firstly, the communication MSE $\norm{\ee_c(k)}^2$ is given by
\begin{small}
	\begin{align}
	\E\left[\norm{\ee_c(k)}^2\right]=&\E \bigg[\bigg\|\frac{1}{|\cM_k|}\sum_{m\in \cM_k}\nabla f_{m,k}-\hat{\gg}(k)\bigg\|^2\bigg]\notag
	\end{align}
\end{small}

Plugging \eqref{gra_diver}, \eqref{de_agg}, \eqref{normalization}, \eqref{module}, \eqref{denoising} and \eqref{demodule} into the above equality, we obtain \eqref{com_error}.

We use $\cP_k$ as the complement of $\cM_k$, i.e., $\cP_k\cup \cM_k=\cM$. For device selection MSE $\norm{\ee_d(k)}^2$, we have
	\begin{small}
		\begin{align}
			&\E \left[\norm{\ee_d(k)}^2\right]\notag\\
			=&\E \Big[ \Big\|{\frac{1}{M}\sum_{m\in \cM}\nabla f_{m,k}-\frac{1}{|\cM_k|}\sum_{m\in \cM_k}\nabla  f_{m,k}}\Big\|^2\Big]\notag\\
			=&\E \Big[\frac{M-|\cM_k|}{M|\cM_k|}\Big(\Big\|{\sum_{m\in \cM_k}\!\left(\nabla f_{m,k}\!-\!\nabla F_{m,k}\right)}\bigg\|\!+\!\Big\|\!{\sum_{m\in \cM_k}\!\nabla F_{m,k}}\Big\|\Big)\notag\\&+\frac{1}{M}\Big\|{\sum_{m \in \cP_k}\left(\nabla f_{m,k}-\nabla F_{m,k}\right)}\bigg\|+\frac{1}{M}\Big\|{\sum_{m \in \cP_k}\!\nabla F_{m,k}}\Big\|\Big]^2\notag\\
			\overset{(a)}{\leq}&\E\Big[\frac{2(M\!\!-\!\!|\cM_k|)}{M}\delta\!+\!\frac{M\!\!-\!\!|\cM_k|}{M|\cM_k|}\Big\|\!\!{\sum_{m\in \cM_k}\!\!\!\nabla F_{m,k}}\Big\|\!\!+\!\!\frac{1}{M}\Big\|\!\!{\sum_{m \in \cP_k}\!\!\nabla F_{m,k}}\bigg\|\bigg]^2\notag\\
			\overset{(b)}{\leq}&4\E\Big[\frac{M-|\cM_k|}{M}\delta+\frac{M-|\cM_k|}{M}\sqrt{\alpha_1+\alpha_2\big\|{\nabla F(\xx(k))}\big\|^2}\bigg]^2\notag\\
			\overset{(c)}{\leq}&8\E\Big(\frac{M-|\cM_k|}{M}\Big)^2\left(\delta^2+\alpha_1+\alpha_2\norm{\nabla F(\xx(k))}^2\right)\notag,
		\end{align}
	\end{small}where   $(a)$ is from triangle inequality and Assumption \ref{as5}, $(b)$ is from Assumption \ref{as4}, and $(c)$ is from the  triangle inequality. 

Finally, we  bound the model asynchrony $\norm{\ee_a(k)}^2$.
\begin{small}
	\begin{align*}
		&\E\norm{\ee_a(k)}^2\\=&\E\Big[\Big\|\nabla F(\xx(k))-\frac{1}{M}\sum_{m=1}^{M}\nabla f_{m,k}\Big\|^2\Big]\\
		{\leq}&\frac{1}{M}\sum_{m=1}^{M}\E\left[\norm{(\nabla F(\xx(k))-\nabla f_{m,k})}^2\right]\\
		\overset{(a)}{\leq}&\frac{2}{M}\sum_{m=1}^{M}\E\left[\norm{\nabla F(\xx(k))-\nabla F_{m,k}}^2+\norm{\nabla F_{m,k}-\nabla f_{m,k}}^2\right]\\
		\overset{(b)}{\leq}&\frac{2}{M}\sum_{m=1}^{M}\E\norm{\nabla F(\xx(k))-\nabla F_{m,k}}^2+2\delta^2\\
		\overset{(c)}{\leq}&\frac{2L^2}{M}\sum_{m=1}^{M}\E\norm{\xx(k)-\xx({k-\tau_{m,k}})}^2+2\delta^2
		\\\overset{(d)}{\leq}&2L^2\E\norm{\xx(k)-\xx({k-\tau_{k,\mu}})}^2+2\delta^2,
	\end{align*}	
\end{small}where  $(a)$ is from the triangle inequality, $(b)$ is from Assumption \ref{as5}, $(c)$ is from Assumption \ref{as1}. From Assumption \ref{as3}, since the staleness  is bounded, we can always find a variable $\mu$ satisfying  $\mu=\arg\max_{m\in [M]}\norm{x(k)-\xx({k-\tau_{m,k}})}^2$ and hence the step $(d)$ holds. It then follows that
\begin{small}
	\begin{align*}
		&L^2\E \norm{\xx(k)-\xx({k-\tau_{k,\mu}})}^2\\
		=&L^2\E \Big\|{\sum_{j=k-\tau_{k,\mu}}^{k-1}(\xx({j+1})-\xx(j))}\Big\|^2\\
		=&\E \Big\|\!\!{\sum_{j=k-\tau_{k,\mu}}^{k-1}\!\!\!\!\!\!\Big(\hat{\gg}(j)-\frac{1}{|\cM_j|}\!\!\sum_{m\in \cM_j}\nabla F_{m,j}\Big)\!+\!\!\!\!\!\!\!\!\sum_{j=k-\tau_{k,\mu}}^{k-1}\!\!\!\!\frac{1}{|\cM_j|}\!\!\sum_{m\in \cM_j}\!\!\!\nabla F_{m,j}}\Big\|^2\\
		=&\E \Big\|\sum_{j=k-\tau_{k,\mu}}^{k-1}\frac{1}{|\cM_j|}\sum_{m\in \cM_j}\left(\nabla f_{m,j}-\nabla F_{m,j}\right)\\&+\sum_{j=k-\tau_{k,\mu}}^{k-1}\frac{1}{|\cM_j|}\sum_{m\in \cM_j}\nabla F_{m,j}-\!\!\sum_{j=k-\tau_{k,\mu}}^{k-1}\ee_c(j)\Big\|^2\\
		\overset{(a)}{\leq}& 3\E{\Big\|\sum_{j=k-\tau_{k,\mu}}^{k-1}\frac{1}{|\cM_j|}\sum_{m\in \cM_j}\left(\nabla f_{m,j}-\nabla F_{m,j}\right)\!\Big\|^2}\\&+\!3\E{\Big\|\sum_{j=k-\tau_{k,\mu}}^{k-1}\frac{1}{|\cM_j|}\sum_{m\in \cM_j}\nabla F_{m,j}\Big\|^2}+3\E\Big\|\sum_{j=k-\tau_{k,\mu}}^{k-1}\ee_c(j)\Big\|^2\\
		\overset{(b)}{\leq}&3S\E\sum_{j=k-S}^{k-1}\frac{1}{|\cM_j|}\sum_{m\in \cM_j}\norm{\nabla f_{m,j}-\nabla F_{m,j}}^2\\&+3S\E \sum_{j=k-S}^{k-1}\frac{1}{|\cM_j|^2}\Big\|\sum_{m\in \cM_j}\nabla F_{m,j}\Big\|^2+3S\sum_{j=k-S}^{k-1}\E\norm{\ee_c(j)}^2\\
		\overset{(c)}{\leq}& 3S^2\delta^2+3T\E \sum_{j=k-S}^{k-1}\frac{1}{|\cM_j|}\sum_{m\in \cM_j}\norm{\nabla F_{m,j}}^2+3S\!\!\sum_{j=k-S}^{k-1}\!\!\E\norm{\ee_c(j)}^2\\
		\overset{(d)}{\leq}&3S^2\left(\delta^2+\alpha_1+\alpha_2\norm{\nabla F(\xx(k))}^2\right)+3S\E\sum_{j=k-S}^{k-1}\norm{\ee_c(j)}^2,
	\end{align*}
\end{small}where $(a)$ is from $\norm{\sum_{i=1}^{n}\aa_i}^2\leq \sum_{i=1}^{n}\norm{\aa_i}^2,\forall \aa_i\in \mathbb{R}^d$, $(b)$ is from Assumption \ref{as3}, $(c)$ is from Assumption \ref{as5}, and $(d)$ is from Assumption \ref{as4}. Therefore,  we have
\begin{small}
	\begin{align*}
		\E \norm{\ee_a(k)}^2&\leq 6S^2\left(\delta^2+\alpha_1+\alpha_2\E\norm{\nabla F(\xx(k))}^2\right)\\&+6S\E\sum_{j=k-S}^{k-1}\!\!\!\!\norm{\ee_c(j)}^2+2\delta^2,
	\end{align*}
\end{small}which completes the proof of Proposition \ref{pro1}.
	\section{Proof of Corollary \ref{coro1}}\label{app_b}
	We expand the communication MSE in \eqref{com_error} and obtain
	\begin{small}
		\begin{align}
		&\E\norm{\ee_c(k)}^2 
		\overset{(a)}{=}\sum_{c=1}^{C}\E\Big[\sum_{m\in \cM_k}2\Delta_m(k)^2+\zeta(k)^2\delta^2_n\!\Big]\label{aa_zeta},
		\end{align}
	\end{small}where $	\Delta_m(k) = \frac{\sqrt{v_m{(k)}}}{|\cM_k|} - \zeta(k)|h_m(k)|b_m(k)$, and $(a)$ is from $\E[\nn(k)]=\mathbf{0}$ and $\E[r_m^{(k)}(c)r_m^{(k)}(c)^\dagger]=2$. Note that \eqref{aa_zeta} is a convex quadratic function with respect to $\zeta(k)$, and the optimal $\zeta(k)$ minimizing \eqref{aa_zeta} is given by
	\begin{small}
		\begin{align}
			&\zeta(k)=\frac{\sum_{m\in \cM_k}\frac{\sqrt{v_m{(k)}}}{|\cM_k|}|h_m(k)|b_m(k)}{\sum_{m\in \cM_k}|h_m(k)|^2b_m(k)^2+\delta^2_n/2}\label{zeta}.
		\end{align}
	\end{small}

	Substituting  \eqref{zeta} into \eqref{aa_zeta}, we obtain 
			\begin{small}
		\begin{align}
			\E\norm{\ee_{c}(k)}^2&=
			2\sum_{c=1}^{C}\E\Bigg[\sum_{m\in \cM_k}\frac{{v_m{(k)}}}{|\cM_k|^2}-\varphi(k)\zeta(k)\Bigg]\label{50},
			\end{align}
	\end{small}where $\varphi(k)=\sum_{m\in \cM_k}\frac{\sqrt{v_m{(k)}}}{|\cM_k|}|h_m(k)|b_m(k)$. In addition,  \eqref{50} can be further bounded as 
	\begin{small}
	\begin{align}
			\E\norm{\ee_{c}(k)}^2\!&\overset{(a)}{=}\frac{2C{v_m(k)}}{|\cM_k|^2}\E\Big[|\cM_k|-\varepsilon(k)\Big]
			\notag\\&\overset{(b)}{\leq}{2C{v_m(k)}}\E\Big[|\cM_k|-\varepsilon(k)\Big]\notag\\
			&\overset{(c)}{\leq}\norm{\nabla f_m(\xx(k-\tau_{m,k}))}^2\E\Big[|\cM_k|-\varepsilon(k)\Big]
			\notag\\&\overset{(d)}{\leq} \!\left(\delta^2+\alpha_1+\alpha_2\norm{\nabla F(\xx(k))}^2\right)\Big[|\cM_k|-\varepsilon(k)\Big],\notag
		\end{align}
	\end{small}where $\varepsilon(k)=\frac{(\sum_{m\in \cM_k}|h_m(k)|b_m(k))^2}{\sum_{m\in \cM_k}|h_m(k)|^2b_m(k)^2+\delta^2_n/2}$, step $(a)$ is based on the assumption that the gradient variance of each device is equal, i.e., $v_1^{(k)}=v_2^{(k)}=\cdots= v_m^{(k)}$, step $(b)$ is based on the fact that when  $\E\norm{\ee_c(k)}\neq0$, $|\cM_k|\geq1$, step $(c)$ is from ${	\sum_{d=1}^{D}({g}_m^{(k)}(d)-\bar{{g}}_m^{(k)})^2}=\sum_{d=1}^{D}((g_m^{(k)}(d))^2-(\bar{g}_m^{(k)})^2)\leq \sum_{d=1}^{D}(g_m^{(k)}(d))^2=\norm{\nabla f_m(\xx(k-\tau_{m,k}))}^2$, and step $(d)$ is from Assumption \ref{as4} and \ref{as5}. Thus, we  obtain \eqref{c_mse}.
	\section{Proof of Proposition \ref{pro_2}}\label{app_c}	
	As demonstrated in \cite{friedlander2012hybrid},   Assumptions \ref{as1} and \ref{as2} can provide an upper bound of the loss function $F(\xx(k+1))$ with respect to
	the recursion \eqref{update_error}, as shown in the following lemma.
	\begin{lemma}\label{lemma1}\cite[Lemma 2.1]{friedlander2012hybrid}
		Under Assumptions \ref{as1} and \ref{as2}, with learning rate $\lambda=1/L$, we have
		\begin{small}
			\begin{align}
				\!\!\!\!\E\left[F(\xx(k+1))\right] \!\leq\! \E [F(\xx(k)) \! \! -\! \!  \frac{1}{2L}\!  \!  \norm{\nabla F(\xx(k))}^2\!\!\!+\!\!\frac{1}{2L}\!\! \norm{\ee(k)}^2],\label{error_conver}\!\!\!\!
			\end{align}
		\end{small}where the expectations are taken over the communication noise.
	\end{lemma}
	Then, we bound the overall MSE  $\E\norm{\ee(k)}^2$ in \eqref{error_conver} as\begin{small}
		\begin{align}
	\E	\norm{\ee(k)}^2\overset{(a)}{\leq}\notag& 3\E[ \norm{\ee_c(k)}^2]+3\E[\norm{\ee_d(k)}^2]+3\E[\norm{\ee_a(k)}^2]\notag\\
	\overset{(b)}{\leq}&\alpha_2g(\qq_k,\aa_k,\bb_k)\norm{\nabla F(\xx(k))}^2\notag\\&\quad+(\delta^2+\alpha_1)g(\qq_k,\aa_k,\bb_k)+6\delta^2\label{e_over},
		\end{align}
\end{small}where $\aa_k \in \mathbb{R}^M$, $\bb_k \in \mathbb{R}^M$ and $\qq_k \in \mathbb{R}^3$ be the $k$-th column of $\mA$, $\mB$ and $\mQ$, respectively, $g(\qq_k,\aa_k,\bb_k)= 18S^2+24(\frac{M-\sum_{m=1}^{M}a_m(k)}{M})^2+18S\sum_{j=k-S}^{k}(\sum_{m=1}^{M}a_m(j)-\frac{(\sum_{m=1}^{M}a_m(j)\sqrt{g_0}\norm{\qq(j)-\ww_m}^{-1}b_m(j))^2}{\sum_{m=1}^{M}a_m(j){g_0\norm{\qq(j)-\ww_m}^{-2}}b_m(j)^2+\delta^2_n/2})$, $(a)$ is from the triangle inequality, and $(b)$ is based on  \eqref{de_sel_error}, \eqref{asy_error} and \eqref{c_mse},
	
	Substituting \eqref{e_over} into \eqref{error_conver}, and using  $\norm{\nabla F(\xx(k))}^2\geq2\mu\left(F(\xx(k))-F(\xx^\star)\right)$ (from Assumption \ref{as2}), we have
	\begin{small}
	\begin{align}
		&F(\xx({k+1}))-F(\xx^\star)\leq{\frac{1}{2L}\left[g(\qq_k,\aa_k,\bb_k)(\delta^2+\alpha_1)+6\delta^2\right]}\notag\\&+{\big(1-\frac{\mu}{L}\left[1-\alpha_2g(\qq_k,\aa_k,\bb_k)\right]\big)}\left(F(\xx(k))-F(\xx^\star)\right).\label{fin_con}
	\end{align}
\end{small}

Recursively applying \eqref{fin_con} for $K$ times, we obtain \eqref{pro2}.
	
	\section{Proof of Proposition \ref{propo4}}\label{app_d}
	We introduce auxiliary variable $\ff\in\mathbb{R}^K$ satisfying $f_k\geq g(\aa_k),\forall k$, and $\dd\in\mathbb{R}^K$ satisfying $d_k\leq \frac{(\sum_{m=1}^{M}a_m(k)|h_m(k)|b_m(k))^2}{\sum_{m=1}^{M}a_m(k)|h_m(k)|^2b_m(k)^2+\delta^2_n/2}, \forall k$ and reformulate \eqref{recastpro3} as
	\begin{small}
		\begin{subequations}\label{recastpro5}
			\begin{align}
				&\min_{\mA,\pp,\dd,\boldsymbol{\omega},\yy} \quad\log \left(\sum_{k=1}^{K}e^{\cc_k^\T \yy}\right)+\frac{1}{\eta}\sum_{m=1}^{M}P(\mA,\bar{\mA})\\
				\text{s.t.}
				&\quad 		f_k \geq g(\aa_k),\forall k,\label{46b}\\
				&\quad  d_k\leq \frac{(\sum_{m=1}^{M}a_m(k)|h_m(k)|b_m(k))^2}{\sum_{m=1}^{M}a_m(k)|h_m(k)|^2b_m(k)^2+\delta^2_n/2}, \forall k,\label{dc_cons1}\\
					&\quad 		f_k\leq \min\Big(\frac{2Lp_i-6\delta^2}{\delta^2+\alpha_1},\frac{Lp_{K+k}+\mu-L}{\mu\alpha_2} \Big),\forall k,\\
				&\quad 	 \eqref{staleness},\eqref{active},\eqref{33d}.
			\end{align}
		\end{subequations}
	\end{small}

 When \eqref{recastpro5} is minimized,  at the optimal point 
 \eqref{46b} and \eqref{dc_cons1} must be held with equality; hence problem \eqref{recastpro5} is equivalent  to problem \eqref{recastpro3}.
	
Problem \eqref{recastpro5} now includes non-convex constraints \eqref{33d} and \eqref{dc_cons1}. 
Denote $\{\mA^t,\yy^t\}$ as the  optimizer  $\{\mA,\yy\}$ obtained at the previous inner iteration. 
Note that \eqref{33d} is a constraint where a linear function is less than a convex function. Since the  first-order Taylor expansion serves as a global lower bound of any convex function, \eqref{33d}  can be  convexified with a surrogate function by taking the  first-order Taylor expansion to its right-hand side (RHS).
We thus obtain \eqref{exp_y}.

Furthermore, the RHS of \eqref{dc_cons1} is convex because it is  a convex quadratic-over-linear function under a composition with affine mapping \cite[Section 3.2]{boyd2004convex}. Similarly,
 \eqref{dc_cons1} can be convexified by taking the first-order Taylor expansion to the RHS, as shown in \eqref{43_cpms}.
 
We now complete the convexification of all  non-convex constraints in problem \eqref{recastpro3}. The transformed problem is given in \eqref{recastpro4}. Note that all the non-convex functions and their surrogate functions (derived via first-order Taylor expansion) have identical values and gradients at local point $\{\mA^t,\yy^t\}$.  Moreover,  the transformed problem satisfies the Slater's condition. According to \cite[Section III-B]{10100674}, the solution of \eqref{recastpro4} terminates  at a KKT point of \eqref{recastpro3} and qualifies as a local minimum if it resides in the interior of  \eqref{recastpro4}'s feasible set.

	\bibliographystyle{IEEEtran}
	\bibliography{IEEEabrv,mybib}
\end{document}